\documentclass[journal,twoside,web]{ieeecolor}
\usepackage{generic}

\usepackage{cite}

\usepackage{amssymb,amsmath,latexsym,amsfonts,amsthm,mathtools}

\usepackage[colorlinks=true]{hyperref}
\hypersetup{colorlinks, breaklinks, citecolor=blue, linkcolor=blue, urlcolor=blue}
\usepackage{graphicx}
\graphicspath{{figures/}{results_data_informative_safety/}}
\usepackage{float}
\usepackage{caption}
\usepackage[font=scriptsize]{subcaption}
\usepackage{textcomp}
\usepackage{xcolor}
\definecolor{darkpastelpurple}{rgb}{0.59, 0.44, 0.84}
\usepackage[nameinlink]{cleveref}
\Crefformat{figure*}{#2Fig.~#1#3}
\Crefmultiformat{figure*}{Figs.~#2#1#3}{ and~#2#1#3}{, #2#1#3}{ and~#2#1#3}

\usepackage{tikz}
\usetikzlibrary{automata, shapes, arrows, calc, arrows.meta, fit, positioning}

\theoremstyle{plain}
\newtheorem{theorem}{Theorem}
\newtheorem{lemma}{Lemma}
\newtheorem{corollary}{Corollary}
\newtheorem{proposition}{Proposition}
\newtheorem*{problem*}{Problem}
\newtheorem{problem}{Problem}
\Crefname{problem}{Problem}{Problems} 
\crefname{problem}{problem}{problems}     

\theoremstyle{remark}
\newtheorem{remark}{Remark}
\newtheorem{assumption}{Assumption}
\Crefname{assumption}{Assumption}{Assumptions} 
\crefname{assumption}{assumption}{assumptions}     

\theoremstyle{definition}

\usepackage{mathrsfs}

\usepackage{newtxmath}
\usepackage{booktabs}

\newcommand{\R}{\mathbb{R}}

\newcommand{\diag}{\operatorname{diag}}
\newcommand{\one}{\mathbf{1}}
\newcommand{\zero}{\mathbf{0}}
\newcommand{\pos}[1]{\left[#1\right]_{+}}
\newcommand{\cK}{\mathcal{K}}
\newcommand{\cU}{\mathcal{U}}
\newcommand{\cX}{\mathcal{X}}
\newcommand{\cG}{\mathcal{G}}
\newcommand{\cL}{\mathcal{L}}

\begin{document}
\title{Networked Admissibility-Preserving Control for Directed Safe Coordination}
\author{Abhinav Sinha,~\IEEEmembership{Senior Member,~IEEE}, Lohitvel Gopikannan,~\IEEEmembership{Student Member,~IEEE}, and Shashi Ranjan Kumar,~\IEEEmembership{Senior Member,~IEEE}%
\thanks{A. Sinha is with the GALACxIS Lab, Department of Aerospace Engineering, University of Cincinnati, Cincinnati, OH 45221, USA (e-mail: abhinav.sinha@uc.edu). L. Gopikannan and S. R. Kumar are with the Intelligent Systems and Control (ISaC) Lab, Department of Aerospace Engineering, Indian Institute of Technology Bombay, Mumbai 400076, India. (e-mails: lohitvel@aero.iitb.ac.in, srk@aero.iitb.ac.in).}
}
\maketitle

\begin{abstract}
This paper addresses safe coordination of scalar agents with constrained input dynamics over fixed rooted digraphs. We propose a networked Admissibility-Preserving Control (APC) architecture for a common moving output corridor under heterogeneous asymmetric input bounds, wherein the Admissibility-Preserving Input Realization (APIR) generates the physical inputs. In logarithmic barrier coordinates, the closed-loop dynamics form an exact cascade of exponentially decaying realization errors and directed consensus dynamics. We derive direction-specific compatibility conditions that guarantee uniform corridor and input margins, bounded commands, unique complete solutions, and exponential consensus for compact initial sets. The closed-form expression for the agreement value quantifies the effect of the initial APIR mismatch. Under the corresponding compatibility conditions, partial pinning guarantees safe reference tracking whenever the root component contains at least one informed agent. Simulations illustrate the coordination guarantees.
\end{abstract}

\section{Introduction}\label{sec:introduction}
\IEEEPARstart{S}{afety-critical} coordination requires every agent to remain admissible from initialization through agreement. For individual systems, barrier Lyapunov functions enforce state and output constraints \cite{TeeGeTay2009,TeeRenGe2011}, while control barrier functions formalize forward invariance \cite{AmesXuGrizzleTabuada2017}. In multi-agent systems, local interactions determine the collective motion; the unconstrained consensus problem has a mature theory \cite{OlfatiSaberMurray2004,RenBeard2005,OlfatiSaberFaxMurray2007}. Networked barrier-Lyapunov designs extend output-constraint enforcement to consensus \cite{ShenXu2018,TangYuDongLiRen2023}. Although these formulations cover directed interaction and, in some cases, dynamic output bounds with input saturation, merely lifting the barrier construction across the network does not treat the applied input as an independently initialized actuator state driven by a distinct bounded command. The actuator's initial state can alter both the safety transient and the eventual directed agreement value. The output corridor and actuator bounds must remain invariant while the realization error decays.

Over connected undirected graphs, safety-constrained coordination has been studied through consensus tracking with full state and input constraints \cite{FuWenYu2023}, node- and edge-wise funnel coupling \cite{LeeTrennShim2022,LeeBergerTrennShim2023}, and networked APC under asymmetric actuation and time-varying output constraints \cite{SinhaKumar2026}. Rooted, unbalanced digraphs replace average preservation with a topology-dependent collective mode. Under directed interaction, barrier-Lyapunov designs have enforced output constraints \cite{ShenXu2018} and dynamic output bounds under input saturation \cite{TangYuDongLiRen2023}, but the applied input is represented algebraically or through static saturation. It must be noted that a static saturation map clips the requested input instantaneously, whose memoryless representation leaves actuator initialization and the command-to-input transient outside the safety analysis. Actuator-aware barrier formulations have treated dynamically defined inputs \cite{AmesNotomistaWardiEgerstedt2021} and modeled or unmodeled actuator dynamics \cite{HuangChen2021,SeilerJankovicHellstrom2022}. Related multi-agent formulations address safe coordination and navigation under state or actuation constraints, e.g., \cite{NiuAbdallahHayajneh2024,ZinageJhaChandraBakolas2025,WangDongHongJohansson2026}.

Despite these advances, simultaneous forward invariance of a common moving output corridor and heterogeneous asymmetric actuator intervals remains to be established for leaderless consensus over rooted, generally unbalanced digraphs. In parallel, input-constrained consensus has been studied under saturation and event-triggered implementation \cite{YangMengDimarogonasJohansson2014,YiYangWuJohansson2019}; nonconvex velocity and control sets \cite{LinRenYangGui2018}, constraint-governor constructions \cite{OngDjamariHou2020,OngHou2021}, and heterogeneous asymmetric saturation over directed graphs \cite{ZuoJiZhangWangZhang2023,WangZuo2024} have likewise been treated. These designs impose bounds on an algebraic input or reshape admissible references rather than retain the physical input as a differential state with its own initial condition. The network funnel laws \cite{LeeTrennShim2022,LeeBergerTrennShim2023} preserve synchronization envelopes through coupling gains that diverge at the boundary. Although the resulting inputs are shown to remain bounded, the guarantees do not place them within designer-specified actuator intervals or represent them as dynamic actuator states. For single-system tracking, input-constrained funnel control protects a hard input bound by widening the performance funnel during saturation \cite{Berger2024}, thereby relaxing the prescribed error corridor. Safe realization over the network requires the moving output corridor and heterogeneous asymmetric actuator intervals to remain invariant under bounded realization commands. 

Admissibility-Preserving Input Realization (APIR) represents the physical input as an independently initialized actuator state \cite{KumarKumarSinha2026}. Over rooted, generally unbalanced digraphs, root-component APIR transients shift the topology-dependent agreement value and hence the collective trajectory. The safe reference-tracking formulation in \cite{SinhaKumar2026} instead requires reference information at every agent on a connected undirected graph. We retain our APIR mechanism and logarithmic barrier map while developing a leaderless directed synthesis and analyzing the collective and disagreement dynamics. Directional actuator demands motivate compatibility conditions beyond the symmetric inner-bound test in \cite{SinhaKumar2026}. Partial pinning requires graph-theoretic conditions on the informed subset.

To this end, we propose a directed networked Admissibility-Preserving Control (APC) architecture that equips each agent with APIR and regulates the admissible physical-input state through a dedicated command. The contributions can be summarized in four respects: (i) for compatible compact initial sets, the proposed distributed synthesis guarantees unique complete solutions, uniform output and actuator margins, bounded commands, and exponential consensus over rooted, generally non-weight-balanced digraphs using only local, in-neighbor, and common-corridor data online; (ii) we quantify in closed form how the initial APIR mismatch in the root component shifts the directed agreement value; (iii) we show that the one-sided compatibility test admits a strictly larger actuator-bound set than the symmetric test whenever the directional demands differ; and (iv) we establish that partial pinning assigns a prescribed safe trajectory whenever the root component contains an informed agent and the corresponding compatibility conditions hold.

\section{Problem Formulation}\label{sec:problem}
We consider a fixed weighted digraph $\cG=(\mathcal{V},\mathcal{E},\mathbf{A})$ with node set $\mathcal{V}=\{1,\ldots,N\}$. The convention $a_{ij}>0$ means that agent $i$ receives information from agent $j$. The in-degree of node $i$, the in-degree matrix, and the graph Laplacian are $d_i^{\mathrm{in}}=\sum_{j=1}^{N}a_{ij}$, $\mathbf{D}=\diag\{d_1^{\mathrm{in}},\ldots,d_N^{\mathrm{in}}\}$, and $\cL=\mathbf{D}-\mathbf{A}$, respectively.
\begin{assumption}\label{ass:graph}
Digraph $\cG$ contains a directed spanning tree.
\end{assumption}
Under \Cref{ass:graph}, $\cL$ has a simple eigenvalue $\lambda_1(\cL)=0$ and nonzero eigenvalues $\lambda_j(\cL)$, $j=2,\ldots,N$. There exists a unique vector $\boldsymbol{\pi}\in\R_{\geq0}^{N}$ satisfying $\boldsymbol{\pi}^{\top}\cL=\zero^{\top}$ and $\boldsymbol{\pi}^{\top}\one=1$.
We set $\mathbf{P}=\one\boldsymbol{\pi}^{\top}$ and $\boldsymbol{\Pi}=\mathbf{I}_{N}-\mathbf{P}$. The projections $\mathbf{P}$ and $\boldsymbol{\Pi}$ extract the topology-dependent collective mode and the disagreement, respectively. For every rate $0<\lambda_{\cG}<\min_{2\leq j\leq N}\operatorname{Re}(\lambda_j(\cL))$, there is a constant $M_{\cG}\geq1$ such that
\begin{equation}\label{eq:semigroup}
\left\lVert e^{-\cL\tau}-\mathbf{P}\right\rVert
\leq M_{\cG}e^{-\lambda_{\cG}\tau},
~~\tau\geq0.
\end{equation}
The norm in \eqref{eq:semigroup} is Euclidean. The constants in \eqref{eq:semigroup} characterize convergence; safety will be established in the sequel from the row-stochasticity of $e^{-\cL\tau}$. At each node of $\cG$, we consider a scalar agent with dynamics
\begin{equation}\label{eq:plant}
\dot{x}_i=u_i,
~~
y_i=x_i,
\end{equation}
where $u_i$ is the physical input delivered to the plant. The controller generates a dedicated command $v_i$ that drives the APIR
\begin{equation}\label{eq:apir}
\dot{u}_i
=p_{1,i}\left[\sigma_i(u_i)v_i-p_{2,i}u_i\right],
~~p_{1,i},p_{2,i}>0.
\end{equation}
Note that APIR is a dynamic input generator whose independently initialized state $u_i$ is the delivered plant input. The formulation applies to actuators or lower-level input generators admitting \eqref{eq:apir}. The prescribed actuator set is $\cU_i=(\underline{u}_i,\overline{u}_i)$, where $\underline{u}_i<0<\overline{u}_i$, and
\begin{equation}\label{eq:sigma}
\sigma_i(u)=
\begin{cases}
1-\left(u/\overline{u}_i\right)^{\gamma_i},&u\geq0,\\
1-\left(u/\underline{u}_i\right)^{\gamma_i},&u<0,
\end{cases}
~\forall~\gamma_i\in2\mathbb{N}.
\end{equation}
The two branches in \eqref{eq:sigma} agree in value and first derivative at the origin, so $\sigma_i$ is continuously differentiable. At the upper boundary, $\dot{u}_i=-p_{1,i}p_{2,i}\overline{u}_i<0$; at the lower boundary, $\dot{u}_i=-p_{1,i}p_{2,i}\underline{u}_i>0$. Thus, under a bounded command, the APIR vector field points inward at both actuator boundaries and generates a compact invariant interior subset of $\cU_i$ \cite{KumarKumarSinha2026}. The subsequent analysis of the closed-loop system establishes the required command bound and a direction-dependent invariant actuator subset through a self-consistent argument.

Safe coordination also requires each output to remain within a prescribed interval $\cX_i(t)=(\underline{x}_i(t),\overline{x}_i(t))$. The design enforces output safety through the common core $\Omega(t)=\left(\xi_{\ell}(t),\xi_u(t)\right)\subseteq\bigcap_{i=1}^{N}\cX_i(t)$.
These scalar intervals model independent output and input specifications, such as altitude envelopes and individual vertical-velocity limits. Coupled polyhedral state--input constraints require joint feasibility across their facets and fall outside the scope of the present analysis.
\begin{assumption}\label{ass:corridor}
The functions $\xi_{\ell},\xi_u\in C^2$ satisfy, for all $t\geq0$,
\begin{align}
0<\delta_{\xi}
\leq&~h(t):=\xi_u(t)-\xi_{\ell}(t)
\leq\Delta_{\xi},\label{eq:width}\\
|\xi_{\ell}(t)|+|\xi_u(t)|
\leq&~\bar{\xi},
~~
|\ddot{\xi}_{\ell}(t)|+|\ddot{\xi}_u(t)|
\leq\bar{s}_{\xi}.\label{eq:corridor_second}
\end{align}
Known nonnegative constants $\nu_{\ell}^{\pm}$ and $\nu_u^{\pm}$ satisfy
\begin{equation}\label{eq:corridor_velocity}
-\nu_{\ell}^{-}\leq\dot{\xi}_{\ell}(t)\leq\nu_{\ell}^{+},
~
-\nu_u^{-}\leq\dot{\xi}_u(t)\leq\nu_u^{+}.
\end{equation}
\end{assumption}
To represent the moving corridor in the consensus dynamics, we introduce the logarithmic barrier coordinate for $x_i\in\Omega(t)$,
\begin{equation}\label{eq:barrier}
z_i=\ln\!\left(\frac{x_i-\xi_{\ell}(t)}{\xi_u(t)-x_i}\right).
\end{equation}
\begin{lemma}\label{lem:barrier}
For each fixed $t\geq0$, the map in \eqref{eq:barrier} is a $C^1$ diffeomorphism from $\Omega(t)$ onto $\R$, with inverse $x_i=\varphi_t^{-1}(z_i)=\left(\xi_{\ell}(t)+\xi_u(t)e^{z_i}\right)/(1+e^{z_i})$. Let $q_{\ell,i}:=x_i-\xi_{\ell}$ and $q_{u,i}:=\xi_u-x_i$. Along \eqref{eq:plant}, the barrier-coordinate dynamics are $\dot{z}_i=b_i(x_i,t)u_i+\chi_i(x_i,t)$, where $b_i(x_i,t)=h(t)/(q_{\ell,i}q_{u,i})$ and $\chi_i(x_i,t)=-\dot{\xi}_{\ell}/q_{\ell,i}-\dot{\xi}_u/q_{u,i}$. Moreover, the inverse-map derivative satisfies $\left|\partial\varphi_t^{-1}(z)/\partial z\right|\leq\Delta_{\xi}/4$ for $z\in\R$ and $t\geq0$.
\end{lemma}
\begin{proof}
For $x_i\in\Omega(t)$, the gaps $q_{\ell,i}$ and $q_{u,i}$ are positive, while \eqref{eq:width} ensures $h(t)>0$. Hence, differentiating \eqref{eq:barrier} with respect to $x_i$ yields the positive derivative $h(t)/(q_{\ell,i}q_{u,i})$. Further, \eqref{eq:barrier} also implies that $z_i$ tends to $-\infty$ and $+\infty$ at the lower and upper corridor boundaries, respectively. Strict monotonicity and the two endpoint limits establish bijectivity. Solving \eqref{eq:barrier} for $x_i$ yields the inverse map $\varphi_t^{-1}(z_i)=\left(\xi_{\ell}(t)+\xi_u(t)e^{z_i}\right)/(1+e^{z_i})$, whose denominator is positive. Differentiating \eqref{eq:barrier} along \eqref{eq:plant} yields $\dot{z}_i=b_i(x_i,t)u_i+\chi_i(x_i,t)$, where $b_i=h/(q_{\ell,i}q_{u,i})$ and $\chi_i=-\dot{\xi}_{\ell}/q_{\ell,i}-\dot{\xi}_u/q_{u,i}$. Differentiating $\varphi_t^{-1}$ and applying \eqref{eq:width} yield the global estimate $|\partial\varphi_t^{-1}(z)/\partial z|=h(t)e^z/(1+e^z)^2\leq\Delta_{\xi}/4$.
\end{proof}
\begin{problem}\label{prob:main}
Given a compact set $\cK\subset\Omega(0)^N\times\prod_{i=1}^{N}\cU_i$, we seek a distributed command $v_i$ and verifiable APC compatibility conditions such that every solution initialized in $\cK$ is complete\footnote{A solution is complete if its maximal interval of existence is $[0,\infty)$.}, satisfies $x_i(t)\in\Omega(t)$ and $u_i(t)\in\cU_i$ for all $t\geq0$, has bounded commands, and achieves $|x_i(t)-x_j(t)|\to0$.
\end{problem}
Finite actuator authority makes \Cref{prob:main} regional in the sense that APC compatibility identifies compact initial sets whose graph-induced motion and corridor kinematics can be realized within the prescribed input limits.

\section{Main Results}\label{sec:main}
We define the directed consensus field
\begin{equation}\label{eq:beta}
\beta_i=-k\sum_{j=1}^{N}a_{ij}(z_i-z_j),
~~k>0,
\end{equation}
and the APIR error at the level of the realized barrier-coordinate velocity,
\begin{equation}\label{eq:e}
e_i=\dot{z}_i-\beta_i=b_i u_i+\chi_i-\beta_i.
\end{equation}
The error in \eqref{eq:e} aligns the physical realization layer with the network vector field. The commanded signal is chosen as
\begin{equation}\label{eq:controller}
v_i=
\frac{
 b_i p_{1,i}p_{2,i}u_i
 -\dot{b}_i u_i
 -\dot{\chi}_i
 +\dot{\beta}_i
 -c_i e_i
}{b_i p_{1,i}\sigma_i(u_i)},
~~c_i>0.
\end{equation}
Using the boundary gaps introduced in \Cref{lem:barrier}, differentiation of the coefficient functions $b_i$ and $\chi_i$ yields
\begin{align}
\dot{b}_i
=&~b_i\left(
\frac{\dot{h}}{h}
-\frac{u_i-\dot{\xi}_{\ell}}{q_{\ell,i}}
-\frac{\dot{\xi}_u-u_i}{q_{u,i}}
\right),\label{eq:bdot}\\
\dot{\chi}_i
=&~-\frac{\ddot{\xi}_{\ell}}{q_{\ell,i}}
+\frac{\dot{\xi}_{\ell}(u_i-\dot{\xi}_{\ell})}{q_{\ell,i}^{2}}
-\frac{\ddot{\xi}_u}{q_{u,i}}
+\frac{\dot{\xi}_u(\dot{\xi}_u-u_i)}{q_{u,i}^{2}}.\label{eq:chidot}
\end{align}
The derivative of the network field is evaluated without numerical differentiation, that is,
\begin{equation}\label{eq:betadot}
\dot{\beta}_i
=-k\sum_{j=1}^{N}a_{ij}
\left[(b_i u_i+\chi_i)-(b_j u_j+\chi_j)\right].
\end{equation}
Agent $i$ uses its local pair $(x_i,u_i)$, the pairs $(x_j,u_j)$ from its in-neighbors, and corridor function with their derivatives up to second order. Smooth planned corridors, such as altitude envelopes, provide these derivatives as boundary velocities and accelerations. The accelerations enter \eqref{eq:chidot} as feedforward data; we invoke their bounds in \eqref{eq:corridor_second} to establish command boundedness. From \eqref{eq:betadot}, it follows that differentiating neighbor measurements are avoided. Neither $\boldsymbol{\pi}$ nor a graph eigenvalue is needed online.
\begin{proposition}[Exact network--APIR cascade]\label{prop:cascade}
On every interval over which $x_i(t)\in\Omega(t)$ and $u_i(t)\in\cU_i$ for all agents, the closed-loop dynamics generated by \eqref{eq:plant}, \eqref{eq:apir}, and \eqref{eq:controller} satisfy
\begin{equation}\label{eq:cascade}
\dot{\mathbf{z}}=-k\cL\mathbf{z}+\mathbf{e},
~~
\dot{\mathbf{e}}=-\mathbf{C}\mathbf{e};
~~
\mathbf{C}=\diag\{c_1,\ldots,c_N\}.
\end{equation}
\end{proposition}
The cascade in \eqref{eq:cascade} expresses the APC synthesis mechanism, wherein the directed field specifies the coordination objective, while the APIR error dynamics govern its realization through the physical inputs. Our safety analysis establishes sufficient conditions for the feedback inversion in \eqref{eq:controller} to remain uniformly conditioned and for the physical-input demands to lie within the available actuator authority. The compact-set bounds below express these conditions in terms of the initial transformed state and realization mismatch.

For $(\mathbf{x}_0,\mathbf{u}_0)\in\cK$, we obtain $\mathbf{z}_0$ from \eqref{eq:barrier} and define
$\mathbf{e}_0=\mathbf{B}(\mathbf{x}_0,0)\mathbf{u}_0+\boldsymbol{\chi}(\mathbf{x}_0,0)+k\cL\mathbf{z}_0$,
where $\mathbf{B}=\diag\{b_1,\ldots,b_N\}$.
The initial data determine $E_{i,\cK}^{\pm}:=\sup_{\cK}\pos{\pm e_{i,0}}$ and, with $c_{\min}:=\min_i c_i$, the uniform radius
$M_{\cK}:=\sup_{\cK}\|\mathbf{z}_0\|_{\infty}+\left(\sup_{\cK}\|\mathbf{e}_0\|_{\infty}\right)/c_{\min}$.
Let us set $m_{\cK}:=\delta_{\xi}/(1+e^{M_{\cK}})$ and define the directional physical-input demands
\begin{equation}\label{eq:demands}
U_{i,\cK}^{\pm}:=\max\{\nu_{\ell}^{\pm},\nu_u^{\pm}\}+\frac{\Delta_{\xi}}{4}\left(2kd_i^{\mathrm{in}}M_{\cK}+E_{i,\cK}^{\pm}\right).
\end{equation}
These demands follow from row-stochasticity and the corridor kinematics, as shown in the proof of the main result below.
\begin{theorem}[Directed networked APC]\label{thm:main}
Under \Cref{ass:graph,ass:corridor}, if the one-sided APC compatibility conditions
\begin{equation}\label{eq:compatibility}
U_{i,\cK}^{+}<\overline{u}_i,
~~
U_{i,\cK}^{-}<-\underline{u}_i,
~~i\in\mathcal{V},
\end{equation}
hold, then the closed-loop system defined by \eqref{eq:plant}, \eqref{eq:apir}, and \eqref{eq:controller} admits a unique complete solution for every initial condition in $\cK$. The solution satisfies, for all $t\geq0$,
\begin{align}
\xi_{\ell}(t)+m_{\cK}
\leq&~x_i(t)
\leq\xi_u(t)-m_{\cK},\label{eq:output_bound}\\
-U_{i,\cK}^{-}
\leq&~u_i(t)
\leq U_{i,\cK}^{+}.\label{eq:input_bound}
\end{align}
All commands $v_i$ are bounded. Moreover, for every $0<\rho<\min\{k\lambda_{\cG},c_{\min}\}$, there exists $K_{\rho}>0$ such that
\begin{equation}\label{eq:exp_consensus}
\|\mathbf{z}(t)-\one z_{\infty}\|
\leq K_{\rho}e^{-\rho t},
\end{equation}
where
\begin{equation}\label{eq:zinf}
z_{\infty}
=\boldsymbol{\pi}^{\top}\mathbf{z}_0
+\sum_{i=1}^{N}\frac{\pi_i e_{i,0}}{c_i}.
\end{equation}
In physical coordinates,
\begin{equation}\label{eq:physical}
x_i(t)-x_{\mathrm{c}}(t)\longrightarrow0,
~~
x_{\mathrm{c}}(t)
=\frac{\xi_{\ell}(t)+\xi_u(t)e^{z_{\infty}}}{1+e^{z_{\infty}}}.
\end{equation}
Consequently, \eqref{eq:controller} solves \Cref{prob:main} on every compact set satisfying \eqref{eq:compatibility}.
\end{theorem}
\begin{proof}
Let us fix an initial condition in $\cK$. By \Cref{lem:barrier}, \eqref{eq:sigma}, and \Cref{ass:corridor}, the denominator in \eqref{eq:controller} is positive initially and the closed-loop vector field is locally Lipschitz on the admissible open domain. A unique solution exists on a maximal interval $[0,T_{\max})$. Let $T^{\star}$ be the supremum of times for which every $x_i(t)\in\Omega(t)$ and $u_i(t)\in\cU_i$. We derive the following estimates on $[0,T^{\star})$.

\emph{Step 1: Realization-error solution.}
Applying \Cref{prop:cascade} with the initial mismatch $\mathbf{e}_0$, and then using variation of constants for the transformed-state subsystem, yields
\begin{align}
e_i(t)=&e_{i,0}e^{-c_i t};\mathbf{z}(t)=e^{-k\cL t}\mathbf{z}_0
+\int_{0}^{t}e^{-k\cL(t-s)}\mathbf{e}(s)\,ds\label{eq:esolzol}
\end{align}
The first identity preserves each error's sign and implies $\|\mathbf{e}(t)\|_{\infty}\leq\|\mathbf{e}_0\|_{\infty}e^{-c_{\min}t}$.

\emph{Step 2: Uniform transformed-state bound.}
Since $-k\cL$ is Metzler and $\cL\one=\zero$, the semigroup $e^{-k\cL t}$ is nonnegative and row-stochastic, with induced infinity norm one. The identities in \eqref{eq:esolzol} and
$\|e^{-k\cL t}\|_{\infty}=1$ yield, on $[0,T^{\star})$,
\begin{equation}\label{eq:zbound}
\|\mathbf{z}(t)\|_{\infty}
\leq\|\mathbf{z}_0\|_{\infty}+\frac{\|\mathbf{e}_0\|_{\infty}}{c_{\min}}
\leq M_{\cK}.
\end{equation}

\emph{Step 3: Strict output interiority.}
The inverse of \eqref{eq:barrier} yields $x_i-\xi_{\ell}=h(t)e^{z_i}/(1+e^{z_i})$ and $\xi_u-x_i=h(t)/(1+e^{z_i})$. Invoking \eqref{eq:width} and \eqref{eq:zbound} places both distances above $m_{\cK}$ and proves \eqref{eq:output_bound} on $[0,T^{\star})$.

\emph{Step 4: One-sided input bounds.}
By \eqref{eq:beta} and \eqref{eq:zbound}, $|\beta_i(t)|\leq2k d_i^{\mathrm{in}}M_{\cK}$. Solving \eqref{eq:e} with the coefficients in \Cref{lem:barrier} recovers
\begin{equation}\label{eq:u_reconstruction}
u_i=\frac{q_{\ell,i}q_{u,i}}{h}(\beta_i+e_i)
+\frac{q_{u,i}}{h}\dot{\xi}_{\ell}
+\frac{q_{\ell,i}}{h}\dot{\xi}_u.
\end{equation}
The last two terms in \eqref{eq:u_reconstruction} form a convex combination of boundary velocities, bounded in each direction by $\max\{\nu_{\ell}^{\pm},\nu_u^{\pm}\}$ under \eqref{eq:corridor_velocity}. Since $1/b_i=q_{\ell,i}q_{u,i}/h\leq h/4\leq\Delta_{\xi}/4$, the graph bound and sign preservation of $e_i$ in \eqref{eq:esolzol} establish \eqref{eq:input_bound} with the directional demands $U_{i,\cK}^{\pm}$.

\emph{Step 5: APIR conditioning, bounded commands, and completeness.}
Define $\mu_i^{+}:=\overline{u}_i-U_{i,\cK}^{+}$ and $\mu_i^{-}:=-\underline{u}_i-U_{i,\cK}^{-}$. By \eqref{eq:compatibility} and \eqref{eq:input_bound}, both margins are positive and $\underline{u}_i+\mu_i^{-}\leq u_i(t)\leq\overline{u}_i-\mu_i^{+}$. Hence, \eqref{eq:sigma} satisfies
\begin{equation}\label{eq:sigma_lower}
\sigma_i(u_i(t))
\geq\underline{\sigma}_{i,\cK}:=\min\!\left\{
1-\left(\frac{U_{i,\cK}^{+}}{\overline{u}_i}\right)^{\gamma_i},
1-\left(\frac{U_{i,\cK}^{-}}{-\underline{u}_i}\right)^{\gamma_i}
\right\}>0.
\end{equation}
The definition of $b_i$ in \Cref{lem:barrier} yields $4/\Delta_{\xi}\leq b_i\leq\Delta_{\xi}/m_{\cK}^{2}$ under \eqref{eq:width} and \eqref{eq:output_bound}. Under \eqref{eq:corridor_second}, \eqref{eq:corridor_velocity}, \eqref{eq:output_bound}, and \eqref{eq:input_bound}, the identities \eqref{eq:bdot}--\eqref{eq:betadot} make $\dot{b}_i$, $\dot{\chi}_i$, and $\dot{\beta}_i$ bounded. Thus, the numerator of \eqref{eq:controller} is bounded, whereas $p_{1,i}>0$ and the lower bounds on $b_i$ and $\sigma_i$ keep its denominator uniformly positive. Each $v_i$ is bounded on $[0,T^{\star})$.

Additionally, \eqref{eq:corridor_second}, \eqref{eq:output_bound}, and the actuator margins confine $(\mathbf{x},\mathbf{u})$ to a compact subset of the closed-loop domain. The strict margins preclude either boundary at $T^{\star}$, and the continuation theorem yields $T^{\star}=T_{\max}=\infty$.

\emph{Step 6: Directed-network agreement value and exponential consensus.}
We premultiply the state identity in \eqref{eq:esolzol} by $\boldsymbol{\pi}^{\top}$. Since $\boldsymbol{\pi}^{\top}\cL=\zero^{\top}$, integration of the error identity yields
\begin{equation}\label{eq:perron_motion}
\boldsymbol{\pi}^{\top}\mathbf{z}(t)
=\boldsymbol{\pi}^{\top}\mathbf{z}_0
+\sum_{i=1}^{N}\frac{\pi_i e_{i,0}}{c_i}
\left(1-e^{-c_i t}\right).
\end{equation}
Letting $t\to\infty$ in \eqref{eq:perron_motion} identifies \eqref{eq:zinf}. Set $\mathbf{R}(t):=e^{-k\cL t}-\mathbf{P}$. Since $\cL\one=\zero$ and $\mathbf{P}^{2}=\mathbf{P}$, we have $\mathbf{R}(t)\mathbf{P}=\zero$. Subtracting $\one z_{\infty}$ from the state identity in \eqref{eq:esolzol} produces
\begin{align}
\mathbf{z}(t)-\one z_{\infty}
=&~\mathbf{R}(t)\boldsymbol{\Pi}\mathbf{z}_0
+\int_{0}^{t}\mathbf{R}(t-s)\mathbf{e}(s)\,ds\notag\\
&-\one\int_{t}^{\infty}\boldsymbol{\pi}^{\top}\mathbf{e}(s)\,ds.\label{eq:consensus_decomposition}
\end{align}
Denote $\alpha_{\cG}:=k\lambda_{\cG}$ and $E_{\cK}:=\sup_{\cK}\|\mathbf{e}_0\|$. By \eqref{eq:semigroup}, the error identity in \eqref{eq:esolzol}, and compactness of $\cK$, the homogeneous and tail terms in \eqref{eq:consensus_decomposition} decay uniformly at rates $\alpha_{\cG}$ and $c_{\min}$. For $0<\rho<\min\{\alpha_{\cG},c_{\min}\}$, regard $\eta:=\min\{\alpha_{\cG}-\rho,c_{\min}-\rho\}>0$. The convolution term $\int_{0}^{t}\mathbf{R}(t-s)\mathbf{e}(s)\,ds$ in \eqref{eq:consensus_decomposition} has Euclidean norm at most $M_{\cG}E_{\cK}$ times
\begin{equation*}
\int_0^t e^{-\alpha_{\cG}(t-s)}e^{-c_{\min}s}\,ds
\leq e^{-\rho t}t e^{-\eta t}\leq\frac{e^{-\rho t}}{\eta}.
\end{equation*}
These estimates prove \eqref{eq:exp_consensus}. The inverse-map derivative bound in \Cref{lem:barrier} yields $|x_i(t)-x_{\mathrm{c}}(t)|\leq(\Delta_{\xi}/4)|z_i(t)-z_{\infty}|$ and establishes \eqref{eq:physical}.
\end{proof}
The command estimates in the proof of \Cref{thm:main} yield a finite, computable bound $|v_i(t)|\leq\overline{V}_{i,\cK}$. A sufficient condition for a prescribed limit $|v_i|\leq v_{i,\max}$ is $\overline{V}_{i,\cK}\leq v_{i,\max}$, which is not implied by \eqref{eq:compatibility}. Saturating $v_i$ generally changes the exact cascade \eqref{eq:cascade}. If an independently prescribed command limit clips $v_i$, APIR can still preserve the prescribed physical-input interval under the bounded applied command. However, output safety and coordination require a separate analysis beyond the scope of this letter. For fixed demands $U_i^{+},U_i^{-}\geq0$, let $\mathcal{F}_i^{\mathrm{dir}}$ contain actuator pairs satisfying $U_i^{+}<\overline{u}_i$ and $U_i^{-}<-\underline{u}_i$, and let $\mathcal{F}_i^{\mathrm{sym}}$ contain those satisfying $\max\{U_i^{+},U_i^{-}\}<\min\{\overline{u}_i,-\underline{u}_i\}$.
\begin{proposition}\label{prop:set_inclusion}
For every $U_i^{+},U_i^{-}\geq0$, the feasible sets satisfy $\mathcal{F}_i^{\mathrm{sym}}\subseteq\mathcal{F}_i^{\mathrm{dir}}$. The inclusion is strict whenever $U_i^{+}\neq U_i^{-}$.
\end{proposition}
\begin{proof}
The symmetric test implies both directional inequalities. If $U_i^{+}>U_i^{-}$, any actuator pair with $U_i^{-}<-\underline{u}_i\leq U_i^{+}<\overline{u}_i$ passes the directional test and fails the symmetric one. Interchanging the two directions treats $U_i^{-}>U_i^{+}$.
\end{proof}
\begin{remark}\label{rem:tradeoff}
The APIR shift in \eqref{eq:zinf} is $\Delta z_{\infty}:=\sum_i\pi_i e_{i,0}/c_i$ and vanishes when $\mathbf{e}_0=\zero$. For fixed $\mathbf{e}_0$, increasing $c_i$ shrinks its individual contribution, although signed cancellation prevents a monotonic conclusion for $|\Delta z_{\infty}|$. Increasing $c_{\min}$ reduces the mismatch term in $M_{\cK}$ and accelerates realization. The same gains enter the command through $c_i e_i$ in \eqref{eq:controller}. The gain $k$ accelerates nominal disagreement but also changes $\mathbf{e}_0$ through the term $k\cL\mathbf{z}_0$ and enters the demand through $2kd_i^{\mathrm{in}}M_{\cK}$. After changing $k$ or any realization gain $c_i$, the demands $U_{i,\cK}^{\pm}$ must be recomputed to verify the compatibility inequalities in \eqref{eq:compatibility}.
\end{remark}
\begin{remark}\label{rem:root}
Strong connectivity and weight balance are unnecessary in \Cref{thm:main}. If the directed spanning tree has a proper root strongly connected component, then $\pi_i=0$ outside the root component. By \eqref{eq:zinf}, the root component and its APIR transients determine the directed-network agreement value, while every remaining agent converges through the directed spanning tree.
\end{remark}
Under \Cref{thm:main}, the agreement value in \eqref{eq:zinf} depends on the initial transformed state and APIR mismatch. To prescribe this value independently of the initial data, we introduce partial pinning through an informed subset. Let $\mathbf{G}=\diag\{g_1,\ldots,g_N\}\succeq\zero$ and $\mathcal{P}:=\{i:g_i>0\}$. Only agents in $\mathcal{P}$ receive the constant barrier reference $z_{\mathrm{c}}$. To this end, we replace \eqref{eq:beta} by
\begin{equation}\label{eq:pinned_beta}
\beta_i=-k\sum_{j=1}^{N}a_{ij}(z_i-z_j)
-\kappa g_i(z_i-z_{\mathrm{c}}),
~~\kappa>0.
\end{equation}
The physical reference $x_{\mathrm{c}}(t)=\varphi_t^{-1}(z_{\mathrm{c}})$ follows the corridor at a fixed relative position; $z_{\mathrm{c}}=0$ selects its midpoint. We measure tracking by $\boldsymbol{\zeta}:=\mathbf{z}-z_{\mathrm{c}}\one$. By \Cref{lem:barrier}, $|x_i(t)-x_{\mathrm{c}}(t)|\leq(\Delta_{\xi}/4)|\zeta_i(t)|$, so decay of $\zeta_i$ ensures physical reference tracking. The APIR error $\mathbf{e}$ in \eqref{eq:e} instead measures a mismatch in barrier-coordinate velocity. Substitution of \eqref{eq:pinned_beta} into \eqref{eq:e}, with $\cL\one=\zero$, and the cancellation in \eqref{eq:controller} yield
\begin{equation}\label{eq:pinned_reduced}
\dot{\boldsymbol{\zeta}}=-\mathbf{A}_{\mathrm{p}}\boldsymbol{\zeta}+\mathbf{e},
~~\dot{\mathbf{e}}=-\mathbf{C}\mathbf{e},
~~\mathbf{A}_{\mathrm{p}}:=k\cL+\kappa\mathbf{G}.
\end{equation}
Thus, reference assignment becomes stabilization of the origin.

For the pinned field \eqref{eq:pinned_beta}, we evaluate $\mathbf{e}_0$ from \eqref{eq:e} at each initial condition in $\cK$. Define $R_{\infty,\cK}:=\sup_{\cK}\|\boldsymbol{\zeta}_0\|_{\infty}+\sup_{\cK}\|\mathbf{e}_0\|_{\infty}/c_{\min}$, $M_{\cK}^{\mathrm{p}}:=|z_{\mathrm{c}}|+R_{\infty,\cK}$, and $m_{\cK}^{\mathrm{p}}:=\delta_{\xi}/(1+e^{M_{\cK}^{\mathrm{p}}})$. With the signed radii $E_{i,\cK}^{\pm}:=\sup_{\cK}\pos{\pm e_{i,0}}$, the directional physical-input demands are given by
\begin{align}
\hat{U}_{i,\cK}^{\mathrm{p},\pm}
:=&~\max\{\nu_{\ell}^{\pm},\nu_u^{\pm}\}+\frac{\Delta_{\xi}}{4}\left((2kd_i^{\mathrm{in}}+\kappa g_i)R_{\infty,\cK}+E_{i,\cK}^{\pm}\right).
\label{eq:pinned_demand}
\end{align}
Let $\mathcal{R}$ denote the root strongly connected component of $\cG$.
\begin{corollary}[Root-pinned APC]\label{cor:pinning}
Under \Cref{ass:graph,ass:corridor}, suppose $\mathcal{P}\cap\mathcal{R}\neq\varnothing$ and the demands in \eqref{eq:pinned_demand} satisfy $\hat{U}_{i,\cK}^{\mathrm{p},+}<\overline{u}_i$ and $\hat{U}_{i,\cK}^{\mathrm{p},-}<-\underline{u}_i$ for every agent. Then the pinned closed-loop system admits a unique complete solution from every initial condition in $\cK$. For all $t\geq0$, the solution satisfies $\xi_{\ell}(t)+m_{\cK}^{\mathrm{p}}\leq x_i(t)\leq\xi_u(t)-m_{\cK}^{\mathrm{p}}$ and $-\hat{U}_{i,\cK}^{\mathrm{p},-}\leq u_i(t)\leq\hat{U}_{i,\cK}^{\mathrm{p},+}$. All commands $v_i$ are bounded. Both $z_i(t)\to z_{\mathrm{c}}$ and $x_i(t)-x_{\mathrm{c}}(t)\to0$ hold exponentially.
\end{corollary}
\begin{proof}
We order the strongly connected components topologically, with $\mathcal{R}$ first. The matrix $\mathbf{A}_{\mathrm{p}}$ in \eqref{eq:pinned_reduced} becomes block lower triangular. Without root pinning, the root block retains a zero eigenvalue. With root pinning, every diagonal block is an irreducibly diagonally dominant $Z$-matrix. A positive pinning gain makes a row of the root block strictly diagonally dominant, while an incoming edge from another component ensures the same property in each nonroot block. These blocks are nonsingular $M$-matrices, so $-\mathbf{A}_{\mathrm{p}}$ is Hurwitz if and only if $\mathcal{P}\cap\mathcal{R}\neq\varnothing$.

Since $-\mathbf{A}_{\mathrm{p}}$ is Metzler and $\mathbf{A}_{\mathrm{p}}\one=\kappa\mathbf{G}\one\geq\zero$, the semigroup $e^{-\mathbf{A}_{\mathrm{p}}t}$ is nonnegative and substochastic, with $\|e^{-\mathbf{A}_{\mathrm{p}}t}\|_{\infty}\leq1$. Variation of constants in \eqref{eq:pinned_reduced}, with the error identity in \eqref{eq:esolzol}, bounds $\|\boldsymbol{\zeta}(t)\|_{\infty}$ by $R_{\infty,\cK}$ on the maximal admissible interval. Hence, $|z_i(t)|\leq M_{\cK}^{\mathrm{p}}$, and \Cref{lem:barrier} yields the corridor margin $m_{\cK}^{\mathrm{p}}$. Substitution of the field bound $|\beta_i|\leq(2kd_i^{\mathrm{in}}+\kappa g_i)R_{\infty,\cK}$ from \eqref{eq:pinned_beta} and the signed error bounds from \eqref{eq:esolzol} into \eqref{eq:u_reconstruction} establishes $-\hat{U}_{i,\cK}^{\mathrm{p},-}\leq u_i(t)\leq\hat{U}_{i,\cK}^{\mathrm{p},+}$. Differentiating \eqref{eq:pinned_beta} adds $-\kappa g_i\dot{z}_i$ to the right-hand side of \eqref{eq:betadot}. The term $-\kappa g_i\dot{z}_i=-\kappa g_i(\beta_i+e_i)$ is bounded because $\beta_i$ and $e_i$ are bounded. The corridor margin $m_{\cK}^{\mathrm{p}}$ and the strict actuator inequalities in \Cref{cor:pinning} permit the conditioning, bounded-command, and continuation arguments in Step~5 of the proof of \Cref{thm:main}. The diagonal blocks $-\mathbf{A}_{\mathrm{p}}$ and $-\mathbf{C}$ of \eqref{eq:pinned_reduced} are Hurwitz, so $\boldsymbol{\zeta}$ and $\mathbf{e}$ decay exponentially. The inverse-map bound in \Cref{lem:barrier} yields exponential convergence of $x_i(t)-x_{\mathrm{c}}(t)$ to zero.
\end{proof}
Strong connectivity makes $\mathcal{R}=\mathcal{V}$, so any nonempty informed subset suffices in that case. For a general rooted graph, pinning only outside $\mathcal{R}$ cannot stabilize \eqref{eq:pinned_reduced}.

\section{Simulations}\label{sec:numerics}
\begin{figure*}[h!]
    \centering
    \begin{subfigure}[b]{0.32\linewidth}
        \centering
        \resizebox{\linewidth}{!}{%
        \begin{tikzpicture}[
    >=Latex,
    vertex/.style={circle, draw=black, thick, minimum size=8mm, inner sep=0pt, font=\small\bfseries},
    edge/.style={->, thick, shorten >=2pt, shorten <=2pt},
    weight/.style={fill=none, inner sep=1pt, font=\scriptsize}
]
\node[vertex, fill=purple!20] (v5) at (0,0)   {5};
\node[vertex, fill=blue!20]   (v1) at (2.2,0) {1};
\node[vertex, fill=red!20]    (v2) at (4.4,0) {2};
\node[vertex, fill=green!20]  (v3) at (6.6,0) {3};
\node[vertex, fill=orange!25] (v4) at (8.8,0) {4};
\draw[edge] (v5) to[bend left=35] node[weight, above] {$0.7$} (v1);
\draw[edge] (v1) -- node[weight, above] {$1.2$} (v2);
\draw[edge] (v2) -- node[weight, above] {$0.9$} (v3);
\draw[edge] (v3) -- node[weight, above] {$1.4$} (v4);
\draw[edge] (v1) to[bend left=35] node[weight, above] {$0.4$} (v3);
\draw[edge] (v2) to[bend left=35] node[weight, below] {$0.5$} (v5);
\draw[edge] (v4) to[bend left=35] node[weight, below] {$0.8$} (v5);
\end{tikzpicture}%
        }
        \caption{Strongly connected ($\mathcal{G}_1$)}
        \label{fig:strongly_connected1}
    \end{subfigure}
    \hfill
    \begin{subfigure}[b]{0.32\linewidth}
        \centering
        \resizebox{\linewidth}{!}{%
\begin{tikzpicture}[
    >=Latex,
    vertex/.style={circle, draw=black, thick, minimum size=8mm, inner sep=0pt, font=\small\bfseries},
    edge/.style={->, thick, shorten >=2pt, shorten <=2pt},
    weight/.style={fill=none, inner sep=1pt, font=\scriptsize}
]
\node[vertex, fill=purple!20] (v5) at (0,0)   {5};
\node[vertex, fill=blue!20]   (v1) at (2.2,0) {1};
\node[vertex, fill=red!20]    (v2) at (4.4,0) {2};
\node[vertex, fill=green!20]  (v3) at (6.6,0) {3};
\node[vertex, fill=orange!25] (v4) at (8.8,0) {4};
\draw[edge] (v2) to[bend left=35] node[weight, above] {$0.9$} (v1);
\draw[edge] (v1) to[bend left=20] node[weight, above] {$0.8$} (v3);
\draw[edge] (v2) -- node[weight, above] {$1.2$} (v3);
\draw[edge] (v4) -- node[weight, above] {$1.1$} (v3);
\draw[edge] (v2) to[bend right=20] node[weight, below] {$0.8$} (v4);
\draw[edge] (v2) to[bend left=35] node[weight, below] {$0.5$} (v5);
\draw[edge] (v4) to[bend left=35] node[weight, below] {$0.5$} (v5);
\end{tikzpicture}%
        }
        \caption{Directed spanning tree ($\mathcal{G}_2$)}
        \label{fig:spanning_tree}
    \end{subfigure}
    \hfill
    \begin{subfigure}[b]{0.32\linewidth}
        \centering
        \resizebox{\linewidth}{!}{%
\begin{tikzpicture}[
    >=Latex,
    vertex/.style={circle, draw=black, thick, minimum size=8mm, inner sep=0pt, font=\small\bfseries},
    edge/.style={->, thick, shorten >=2pt, shorten <=2pt},
    weight/.style={fill=none, inner sep=1pt, font=\scriptsize}
]
\node[vertex, fill=purple!20] (v5) at (0,0)   {5};
\node[vertex, fill=blue!20]   (v1) at (2.2,0) {1};
\node[vertex, fill=red!20]    (v2) at (4.4,0) {2};
\node[vertex, fill=green!20]  (v3) at (6.6,0) {3};
\node[vertex, fill=orange!25] (v4) at (8.8,0) {4};
\draw[edge] (v1) to[bend left=25] node[weight, above] {$1.6$} (v2);
\draw[edge] (v2) to[bend left=25] node[weight, below] {$1.4$} (v1);
\draw[edge] (v1) to[bend left=40] node[weight, above] {$1.3$} (v3);
\draw[edge] (v3) -- node[weight, above] {$1.1$} (v4);
\draw[edge] (v2) to[bend left=40] node[weight, above] {$0.6$} (v4);
\draw[edge] (v3) to[bend left=35] node[weight, below] {$0.5$} (v5);
\draw[edge] (v4) to[bend left=45] node[weight, below] {$0.9$} (v5);
\end{tikzpicture}%
        }
        \caption{Directed spanning tree ($\mathcal{G}_3$)}
        \label{fig:spanning_tree_G3}
    \end{subfigure}
    \caption{Interaction topologies.}
    \label{fig:graphs}
\end{figure*}
\begin{figure*}[h!]
    \centering
    \begin{subfigure}[t]{0.48\linewidth}
        \centering
        \includegraphics[width=\linewidth]{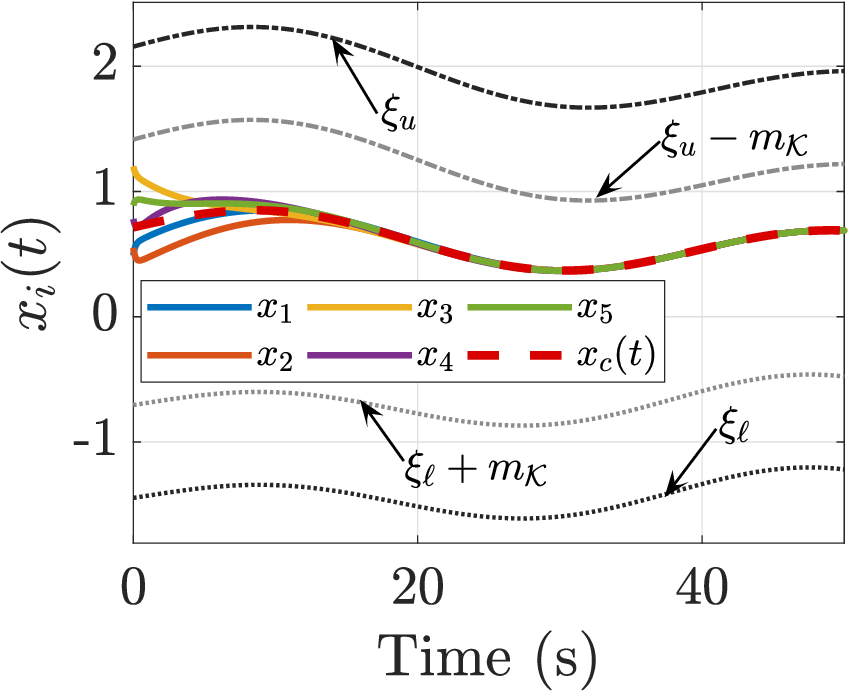}
        \caption{State trajectories $x_i(t)$.}
        \label{fig:x1}
    \end{subfigure}
    \begin{subfigure}[t]{.48\linewidth}
        \centering
        \includegraphics[width=\linewidth]{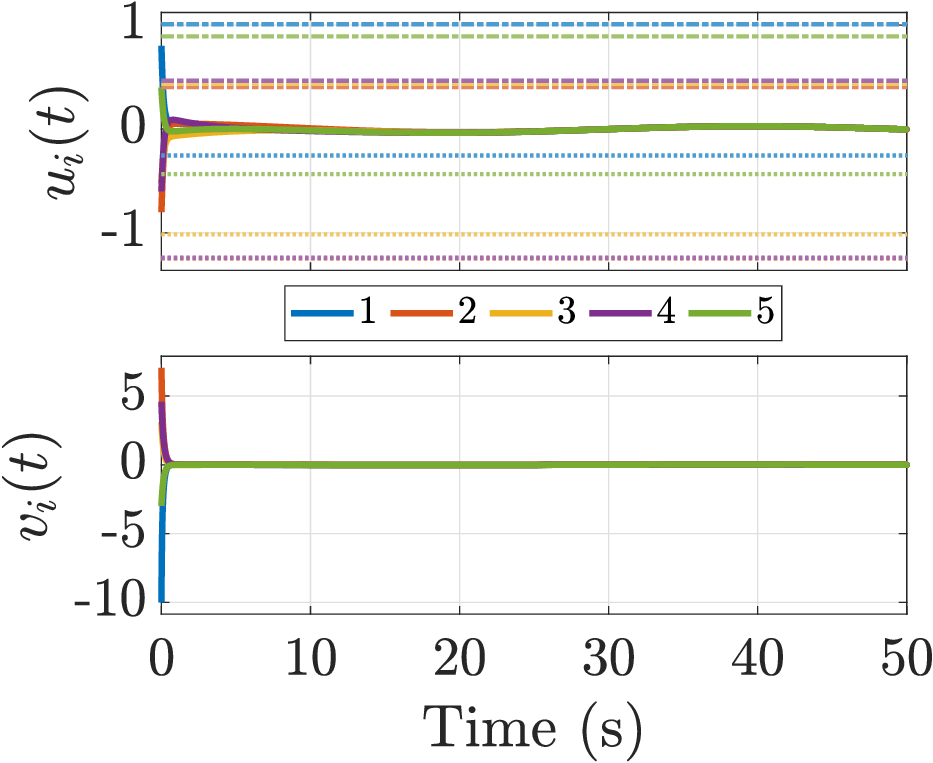}
        \caption{$u_i(t)$ and $v_i(t)$.}
        \label{fig:input_ana1}
    \end{subfigure}
    \caption{Safe consensus under strongly connected digraphs.}
    \label{fig:strongly_connected}
\end{figure*}
\begin{figure*}[h!]
    \centering
    \begin{subfigure}[t]{.48\linewidth}
        \centering
        \includegraphics[width=\linewidth]{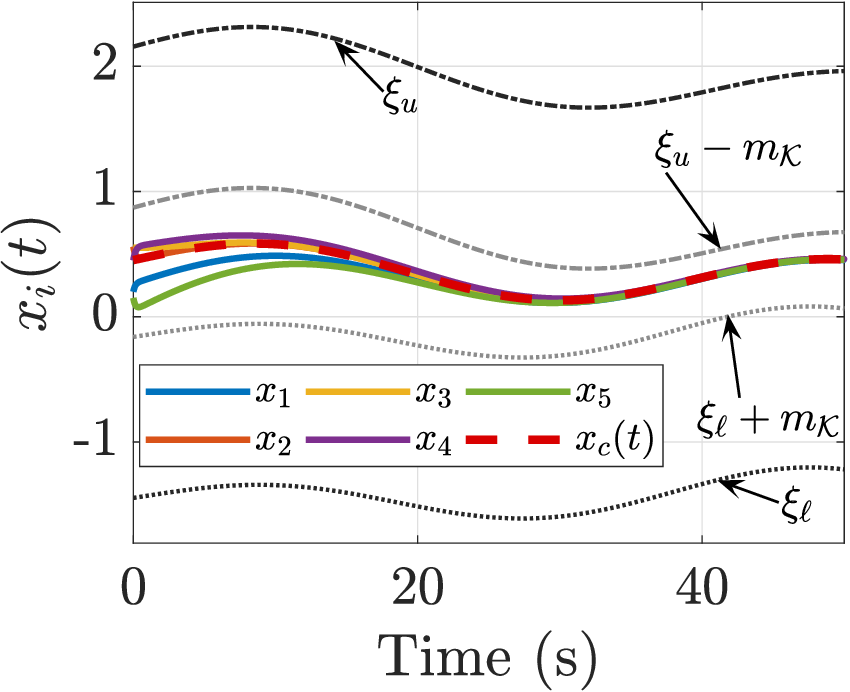}
        \caption{State trajectories $x_i(t)$.}
        \label{fig:x2}
    \end{subfigure}
    \begin{subfigure}[t]{.48\linewidth}
        \centering
        \includegraphics[width=\linewidth]{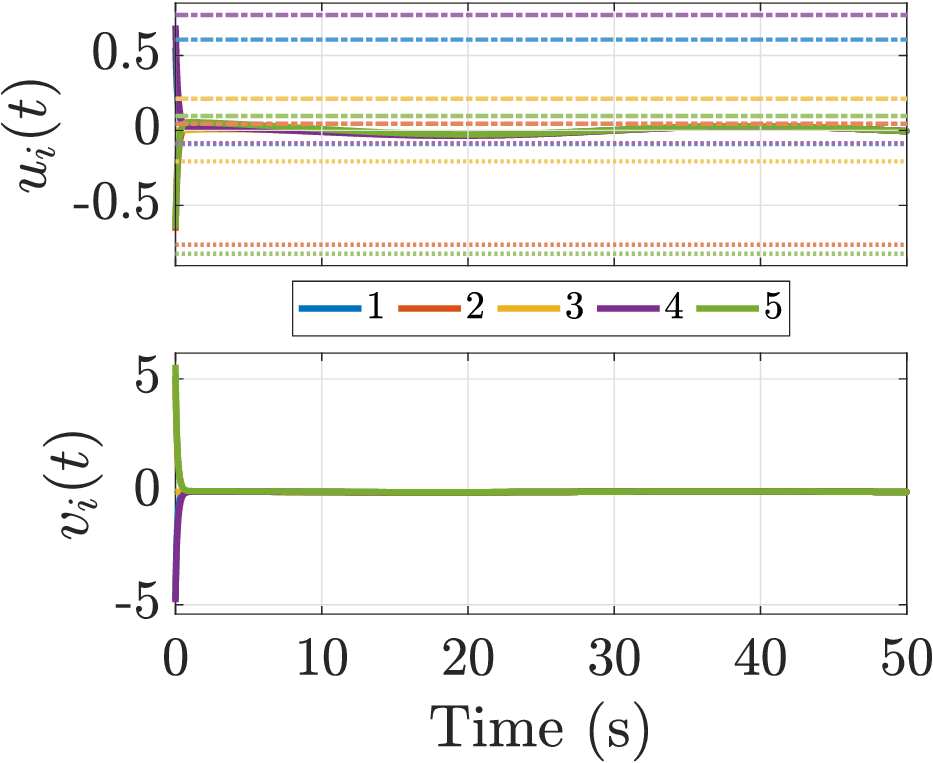}
        \caption{$u_i(t)$ and $v_i(t)$.}
        \label{fig:input_ana2}
    \end{subfigure}
    \caption{Safe consensus under digraphs with only a spanning tree.}
    \label{fig:spanning_tree_results}
\end{figure*}
\begin{figure*}[h!]
    \centering
    \begin{subfigure}[t]{0.48\linewidth}
        \centering
        \includegraphics[width=\linewidth]{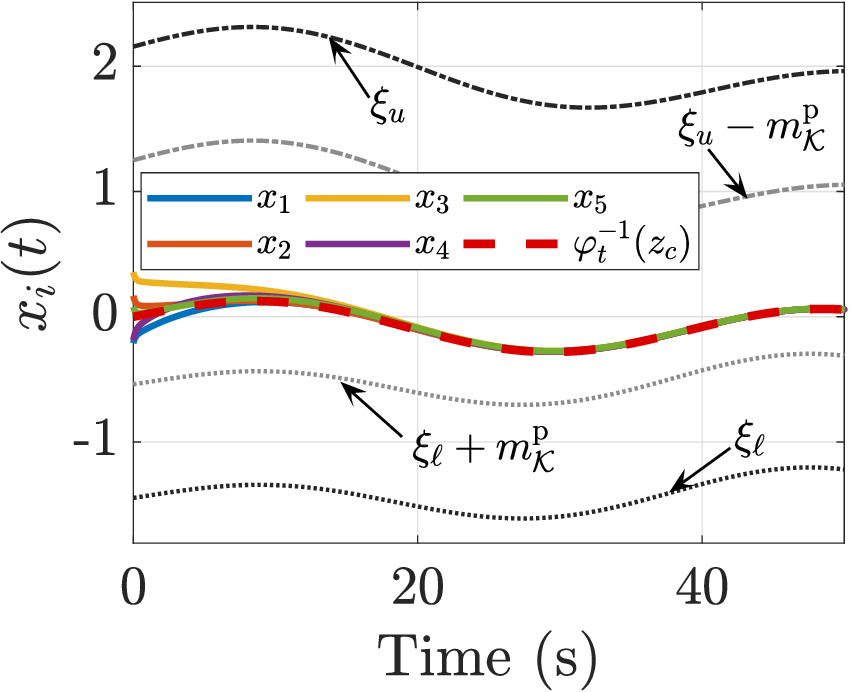}
        \caption{State trajectories $x_i(t)$.}
        \label{fig:x3}
    \end{subfigure}
    \begin{subfigure}[t]{0.48\linewidth}
        \centering
        \includegraphics[width=\linewidth]{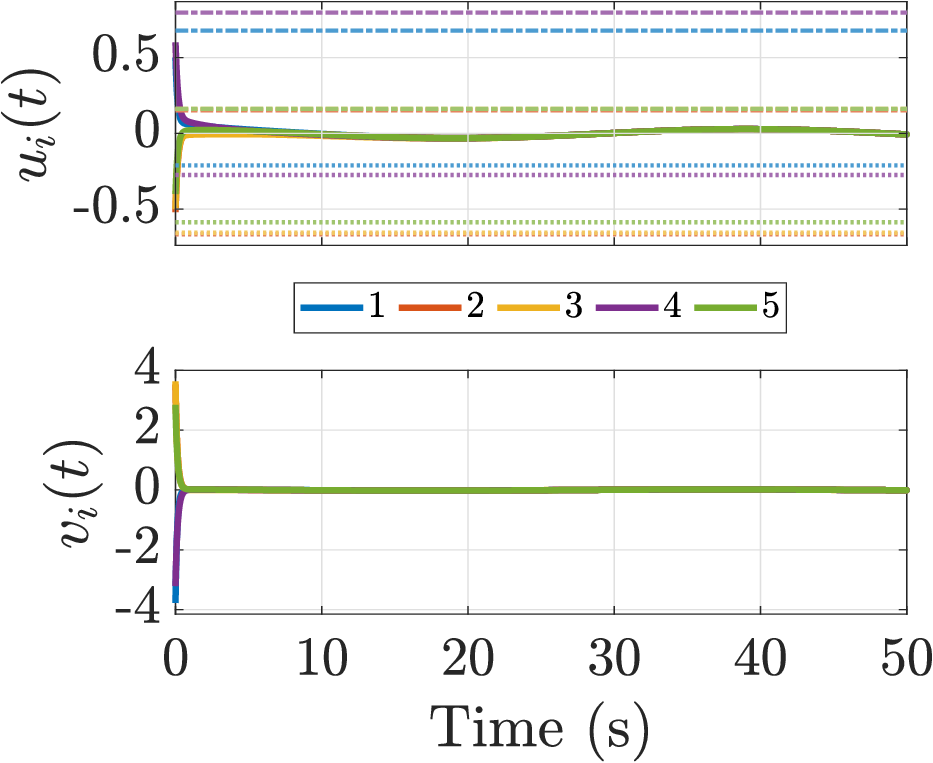}
        \caption{$u_i(t)$ and $v_i(t)$.}
        \label{fig:input_ana3}
    \end{subfigure}
   \caption{Safe consensus under partial pinning.}
    \label{fig:pinning_results}
\end{figure*}
We consider three scenarios with $N=5$ single-integrator agents. In the first, the agents interact over the strongly connected digraph $\mathcal{G}_1$ in \Cref{fig:strongly_connected1}, with $k=0.14$, $p_{1,i}=1$, $\gamma_i=2$, $\mathbf{c}=[8,\,6.5,\,7,\,6.2,\,6.8]^{\top}$, and $\mathbf{p}_2=[0.50,\,0.55,\,0.45,\,0.50,\,0.60]^{\top}$. The heterogeneous, asymmetric actuator intervals are $\mathcal{U}_1=(-0.94,1.2)$, $\mathcal{U}_2=(-1.5,1.36)$, $\mathcal{U}_3=(-1.45,1.54)$, $\mathcal{U}_4=(-1.71,1.56)$, and $\mathcal{U}_5=(-1.45,1.49)$. The initial conditions are $\mathbf{x}(0)=[0.49,\,0.55,\,1.2,\,0.78,\,0.89]^{\top}$ and $\mathbf{u}(0)=[0.8,\,-0.8,\,-0.52,\,-0.6,\,0.4]^{\top}$, and we set $\lambda_{\cG}=0.5$ and $M_{\cG}=1.2$. The corridor has center $c_{\xi}(t)=0.18\sin(0.16t)+0.05\cos(0.05t)$ and half-widths $\Delta_{\ell}(t)=1.45+0.12\sin(0.10t+0.4)$ and $\Delta_u(t)=1.95+0.16\cos(0.08t-0.2)$, so that $\xi_{\ell}=c_{\xi}-\Delta_{\ell}$ and $\xi_u=c_{\xi}+\Delta_u$, with analytic bounds $\delta_{\xi}=3.12$, $\Delta_{\xi}=3.68$, $\nu_{\ell}^{\pm}\approx0.043$, and $\nu_u^{\pm}\approx0.044$.

As shown in \Cref{fig:x1}, all agents remain strictly within the moving corridor and converge to $x_c(t)$, as predicted by \Cref{thm:main}. The actuator margins $(\mu_i^{+},\mu_i^{-})$ are $(0.192,0.686)$, $(0.956,0.253)$, $(1.106,0.436)$, $(1.096,0.472)$, and $(0.599,1.016)$ for agents $1$--$5$, so \eqref{eq:compatibility} holds. \Cref{fig:input_ana1} shows that the commands $v_i(t)$ stay bounded, exceeding the actuator limits during the transient, while the realized inputs remain within $-U_{i,\mathcal{K}}^{-}\le u_i(t)\le U_{i,\mathcal{K}}^{+}$. The predicted agreement value from \eqref{eq:zinf} is $z_{\infty}\approx0.404$, including an APIR shift of approximately $0.001$; at $t=50$~s, $\boldsymbol{\pi}^{\top}\mathbf{z}(50)$ matches it to three decimal places, with $E_z(50):=\max_i|z_i(50)-z_{\infty}|\approx2.62\times10^{-4}$.

The second scenario uses $\mathcal{G}_2$ in \Cref{fig:spanning_tree}. The controller and actuator parameters retain their first-scenario values except $k=0.08$. The initial conditions are $\mathbf{x}(0)=[0.20,\,0.56,\,0.55,\,0.45,\,0.15]^{\top}$ and $\mathbf{u}(0)=[0.55,\,-0.67,\,0,\,0.7,\,-0.65]^{\top}$. For this case, we select $\lambda_{\cG}=0.5$ and $M_{\cG}=2\sqrt{5}$. As shown in \Cref{fig:x2}, the agents again achieve safe consensus, remaining within the moving corridor and converging to $x_c(t)$. \Cref{fig:input_ana2} confirms that the commands $v_i(t)$ remain bounded, and \Cref{fig:input_ana2} verifies that the realized inputs $u_i(t)$ stay within the analytically certified bounds. Here, \eqref{eq:zinf} predicts $z_{\infty}\approx0.108$, including an APIR shift of approximately $-0.120$, which $\boldsymbol{\pi}^{\top}\mathbf{z}(50)$ matches to three decimal places, with $E_z(50)\approx5.71\times10^{-3}$. In both scenarios, \eqref{eq:compatibility} holds with minimum positive/negative actuator margins of approximately $(0.192,0.253)$ and $(0.594,0.627)$, respectively, so the input bounds hold for all time. 

In \Cref{tab:demands}, we compare the directional demands in \eqref{eq:demands} with the observed peaks $u_{i,\mathrm{obs}}^{\pm}:=\max_{0\leq t\leq50}\left[\pm u_i(t)\right]_{+}$, along with the actuator margins $\mu_i^{\pm}$. The analytical bounds guarantee admissibility for every initial condition in $\cK$ and all future time, whereas the observed peaks concern individual trajectories over $[0,50]$~s. The conservatism of the bounds stems from bounding the graph, corridor, and error contributions independently.
\begin{table}[h!]
\centering
\caption{Directional actuator demands, actuator margins, and observed input peaks.}
\label{tab:demands}
\small
\setlength{\tabcolsep}{4pt}
\begin{tabular}{ccrrrrrr}
\hline
Graph & Agent &
$U_{i,\cK}^{+}$ & $u_{i,\mathrm{obs}}^{+}$ & $\mu_i^{+}$ &
$U_{i,\cK}^{-}$ & $u_{i,\mathrm{obs}}^{-}$ & $\mu_i^{-}$\\
\hline
$\mathcal{G}_1$ & 1 & 1.008 & 0.800 & 0.192 & 0.254 & 0.036 & 0.686\\
 & 2 & 0.404 & 0.052 & 0.956 & 1.247 & 0.800 & 0.253\\
 & 3 & 0.434 & 0.027 & 1.106 & 1.014 & 0.520 & 0.436\\
 & 4 & 0.464 & 0.088 & 1.096 & 1.238 & 0.600 & 0.472\\
 & 5 & 0.892 & 0.400 & 0.599 & 0.434 & 0.037 & 1.016\\
\hline
$\mathcal{G}_2$ & 1 & 0.606 & 0.550 & 0.594 & 0.091 & 0.030 & 0.849\\
 & 2 & 0.044 & 0.027 & 1.316 & 0.763 & 0.670 & 0.737\\
 & 3 & 0.211 & 0.027 & 1.329 & 0.207 & 0.033 & 1.244\\
 & 4 & 0.770 & 0.700 & 0.790 & 0.086 & 0.036 & 1.624\\
 & 5 & 0.096 & 0.055 & 1.394 & 0.823 & 0.650 & 0.627\\
\hline
\end{tabular}
\end{table}
To examine the dependence on initialization, we performed $500$ Monte Carlo trials for each of $\mathcal{G}_1$ and $\mathcal{G}_2$, with all other parameters as above. Each component of $\mathbf{z}_0$ and $\mathbf{e}_0$ was sampled independently and uniformly from the axis-aligned hyperrectangle
\begin{equation*}
\mathcal{B}:=\left\{(\mathbf{z}_0,\mathbf{e}_0):\|\mathbf{z}_0-\mathbf{z}_0^{\star}\|_{\infty}\leq0.15,~\|\mathbf{e}_0-\mathbf{e}_0^{\star}\|_{\infty}\leq0.1\right\},
\end{equation*}
where $(\mathbf{z}_0^{\star},\mathbf{e}_0^{\star})$ denotes the initialization of the corresponding scenario. The physical initial states $(\mathbf{x}_0,\mathbf{u}_0)$ follow from the inverse barrier map and $\mathbf{e}_0$, and we denote by $\cK_{\mathrm{MC}}$ the resulting compact set of physical initial conditions.
\begin{figure*}[h!]
    \centering
    \begin{subfigure}[t]{0.48\linewidth}
        \centering
        \includegraphics[width=\linewidth]{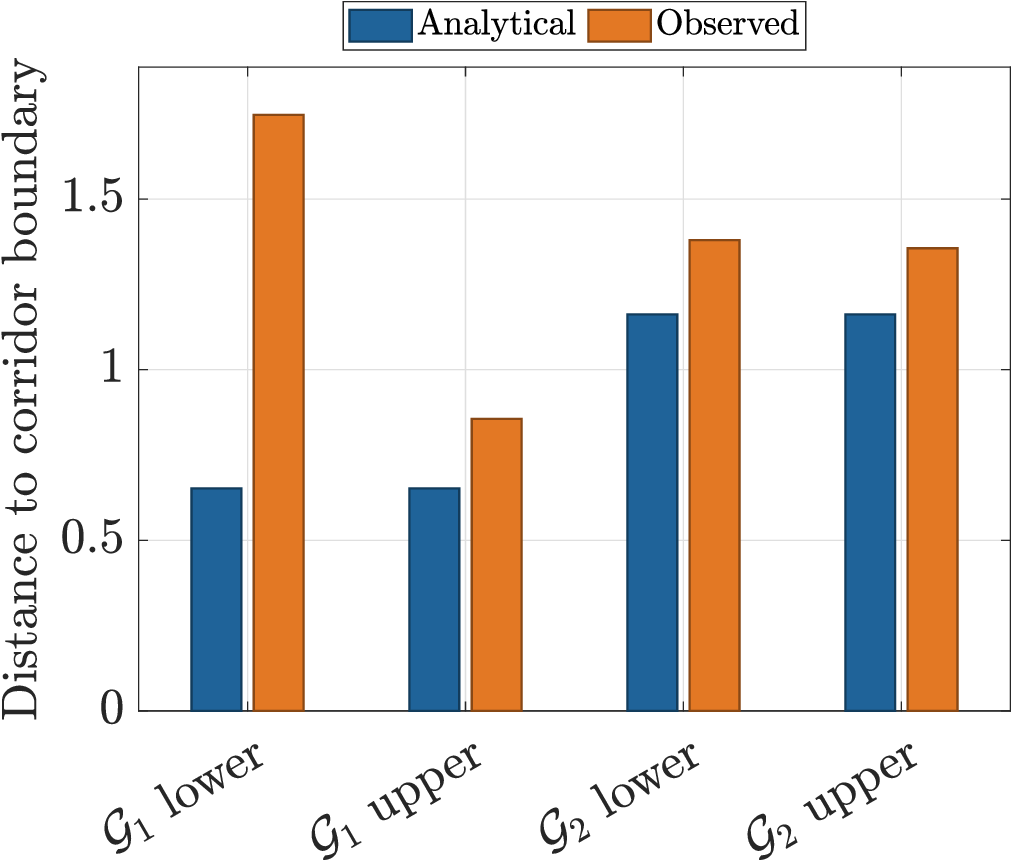}
        \caption{Corridor margins.}
        \label{fig:mc_corridor}
    \end{subfigure}
    \hfill
    \begin{subfigure}[t]{0.48\linewidth}
        \centering
        \includegraphics[width=\linewidth]{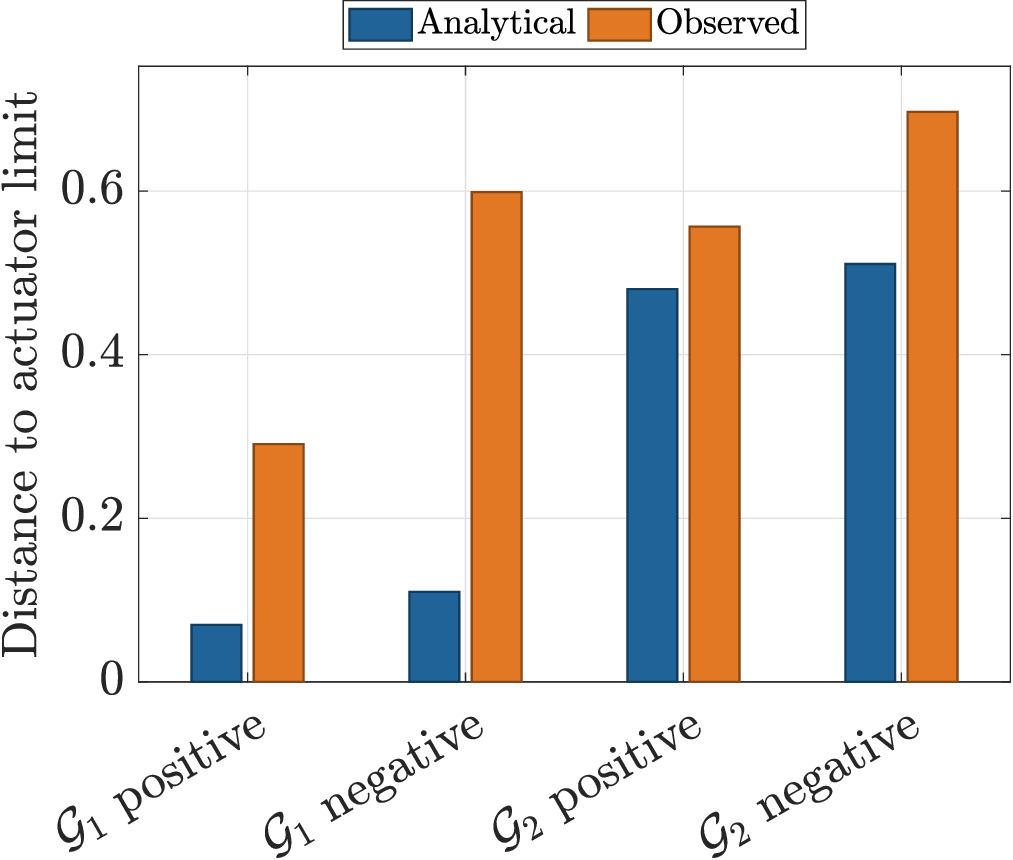}
        \caption{Actuator margins.}
        \label{fig:mc_actuator}
    \end{subfigure}
    \caption{Analytical and minimum observed safety margins for $\mathcal{G}_1$ and $\mathcal{G}_2$ for different initializations.}
    \label{fig:mc_margins}
\end{figure*}
\begin{figure*}[h!]
    \centering
    \begin{subfigure}[t]{0.48\linewidth}
        \centering
        \includegraphics[width=\linewidth]{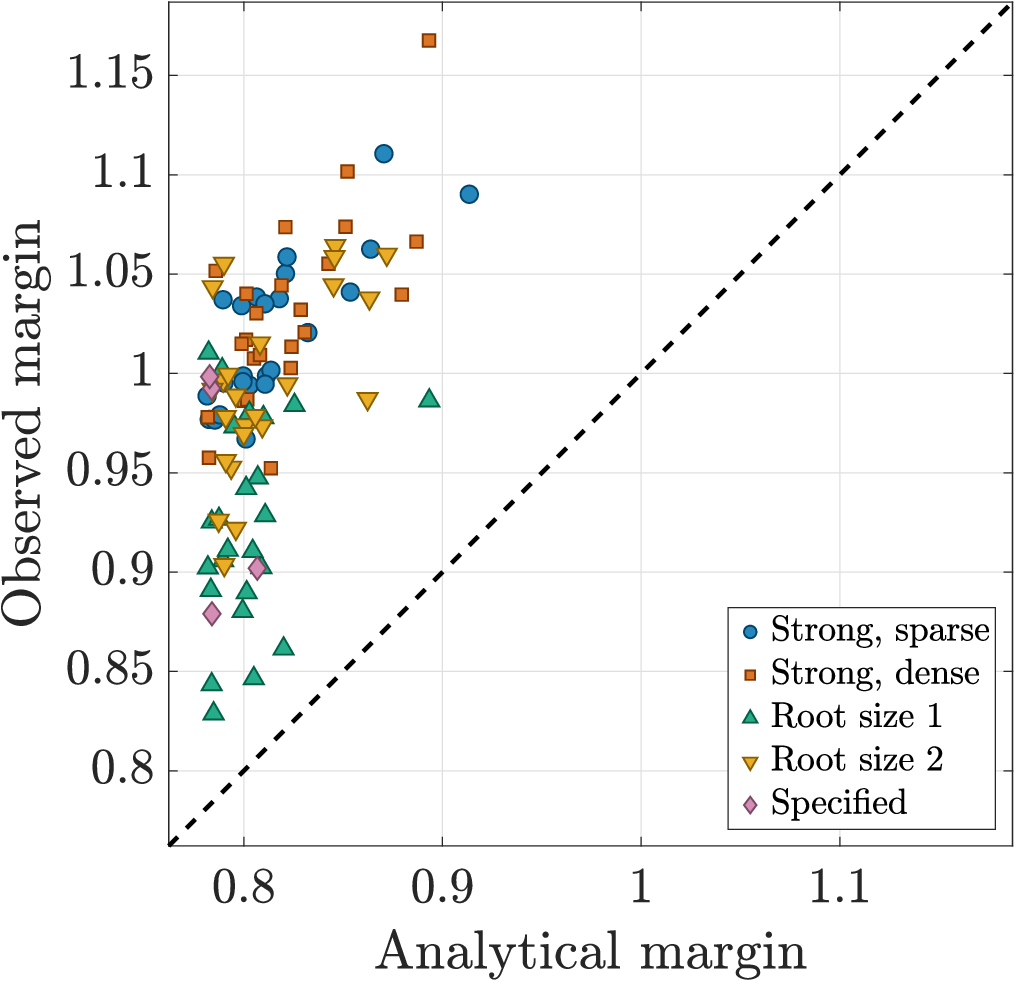}
        \caption{Corridor margins.}
        \label{fig:mc_graph_corridor}
    \end{subfigure}
    \hfill
    \begin{subfigure}[t]{0.48\linewidth}
        \centering
        \includegraphics[width=\linewidth]{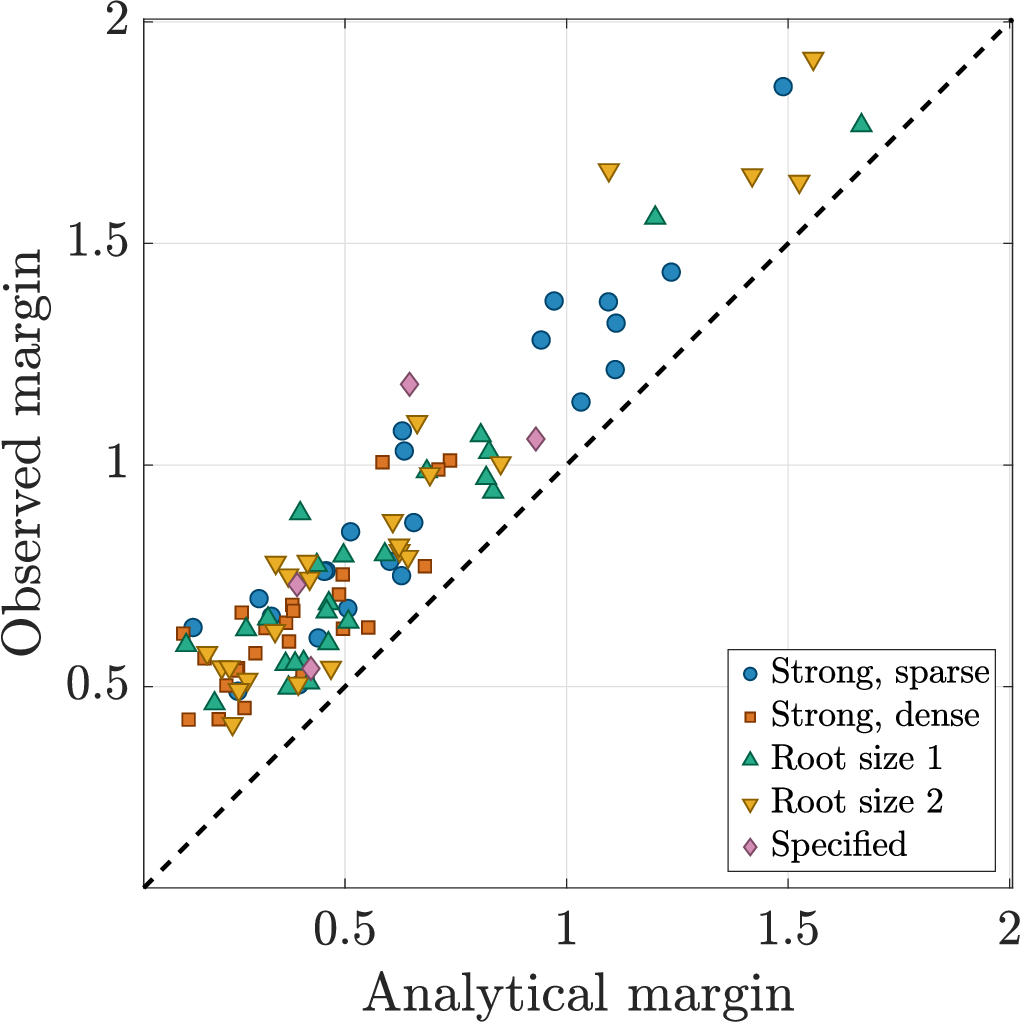}
        \caption{Actuator margins.}
        \label{fig:mc_graph_actuator}
    \end{subfigure}
    \caption{Analytical and minimum observed safety margins for $100$ weighted five-agent digraphs.}
    \label{fig:mc_graph_margins}
\end{figure*}

\Cref{tab:mc_margins} reports the analytical and observed distances to each constraint boundary. The analytical entries follow from evaluating $M_{\cK}$ and \eqref{eq:demands} over $\cK_{\mathrm{MC}}$. The corridor margin is $m_{\cK_{\mathrm{MC}}}$, and the actuator margins are $\min_i\mu_i^{\pm}$, with $\mu_i^{\pm}$ as in $u_i(t)$ (see Step 5 in the proof of \Cref{thm:main}). All analytical margins are positive, so \eqref{eq:compatibility} holds over $\cK_{\mathrm{MC}}$ and \Cref{thm:main} certifies every trial a priori. The observed entries are the minima of $x_i(t)-\xi_{\ell}(t)$, $\xi_u(t)-x_i(t)$, $\overline{u}_i-u_i(t)$, and $u_i(t)-\underline{u}_i$ over all agents, trials, and $t\in[0,50]$~s. All observed margins are positive and exceed the uniform analytical margins, as shown in \Cref{fig:mc_margins}.
\begin{table}[h!]
\centering
\caption{Guaranteed and observed safety margins.}
\label{tab:mc_margins}
\small
\setlength{\tabcolsep}{5pt}
\begin{tabular}{llrrrr}
\toprule
 & & \multicolumn{2}{c}{Corridor} & \multicolumn{2}{c}{Actuator}\\
\cmidrule(lr){3-4}\cmidrule(lr){5-6}
Graph & Evaluation & Lower & Upper & Positive & Negative\\
\midrule
$\mathcal{G}_{1}$ & Analytical & 0.652 & 0.652 & 0.070 & 0.110\\
 & Observed & 1.747 & 0.856 & 0.291 & 0.599\\
$\mathcal{G}_{2}$ & Analytical & 1.162 & 1.162 & 0.480 & 0.511\\
 & Observed & 1.380 & 1.356 & 0.556 & 0.697\\
\bottomrule
\end{tabular}
\end{table}

To examine the dependence on topology, actuator authority, and initialization, we considered an ensemble of $100$ weighted five-agent digraphs. The ensemble comprises $24$ random graphs from each of four families: sparse strongly connected graphs, dense strongly connected graphs, rooted graphs with a singleton root component, and rooted graphs with a two-agent root component. For the strongly connected families, connectivity is ensured by a directed cycle, whereas for the rooted families, an outward spanning tree establishes the root structure; additional edges are then added while preserving the prescribed root component. All edge weights are drawn independently from a log-uniform distribution on $[0.1,2]$. Four additional graphs probe specific structural features: a nonnormal chain, a chain with a repeated defective eigenvalue, a graph containing a weak bridge of weight $0.001$, and a rooted star whose edge weights span $[0.01,100]$.

The corridor and APIR parameters are as above, whereas the actuator limits are randomized, i.e., for each graph and agent, $-\underline{u}_i$ and $\overline{u}_i$ are drawn independently and uniformly from $[0.95,5]$, yielding heterogeneous asymmetric intervals. For each graph, the coordination gain is set to $k=0.14/\max\{1,d_{\max}^{\mathrm{in}}\}$, a center $\mathbf{a}$ is drawn uniformly from $[-0.3,0.3]^5$, and $20$ initial conditions are drawn uniformly from
\begin{equation*}
\mathcal{B}_{\mathrm{G}}:=\left\{(\mathbf{z}_0,\mathbf{e}_0):\|\mathbf{z}_0-\mathbf{a}\|_{\infty}\leq0.7,~\|\mathbf{e}_0\|_{\infty}\leq0.6\right\}.
\end{equation*}
The corresponding physical initial conditions, obtained through the inverse barrier map and $\mathbf{e}_0$, form the set $\cK_{\mathrm{G}}$ of each graph.
Via \Cref{fig:mc_graph_margins}, we compare, for each graph, the analytically certified margin over $\cK_{\mathrm{G}}$ with the minimum observed margin across its $20$ Monte Carlo trajectories, as defined in the preceding study. One may notice that all $100$ points lie above the equality line, indicating that the observed margins consistently exceed their corresponding analytical guarantees. Hence, the certified corridor and input constraints are satisfied for all $2000$ trajectories, covering both strongly connected and rooted topologies, widely varying edge weights, and heterogeneous asymmetric actuator limits.

For partial pinning on $\mathcal{G}_1$, we set $\mathbf{G}=\diag\{1,0,0,1,0\}$, $k=0.1$, $\kappa=0.23$, and $z_{\mathrm{c}}=-0.4$, so only agents 1 and 4 receive the reference. With $\mathbf{z}(0)=[-0.65,-0.21,0,-0.62,-0.31]^{\top}$, i.e., $\mathbf{x}(0)\approx[-0.211,0.167,0.355,-0.186,0.078]^{\top}$, and $\mathbf{u}(0)=[0.50,-0.52,-0.5,0.6,-0.40]^{\top}$, the minimum positive/negative actuator margins are approximately $(0.522,0.729)$, so \Cref{cor:pinning} holds. All agents, including the three uninformed ones, converge to the trajectory generated by $z_{\mathrm{c}}$ (\Cref{fig:x3}), with inputs within the analytical bounds (\Cref{fig:input_ana3}).

\begin{figure*}[h!]
	\centering
	\begin{subfigure}[t]{0.48\linewidth}
		\centering
		\includegraphics[width=\linewidth]{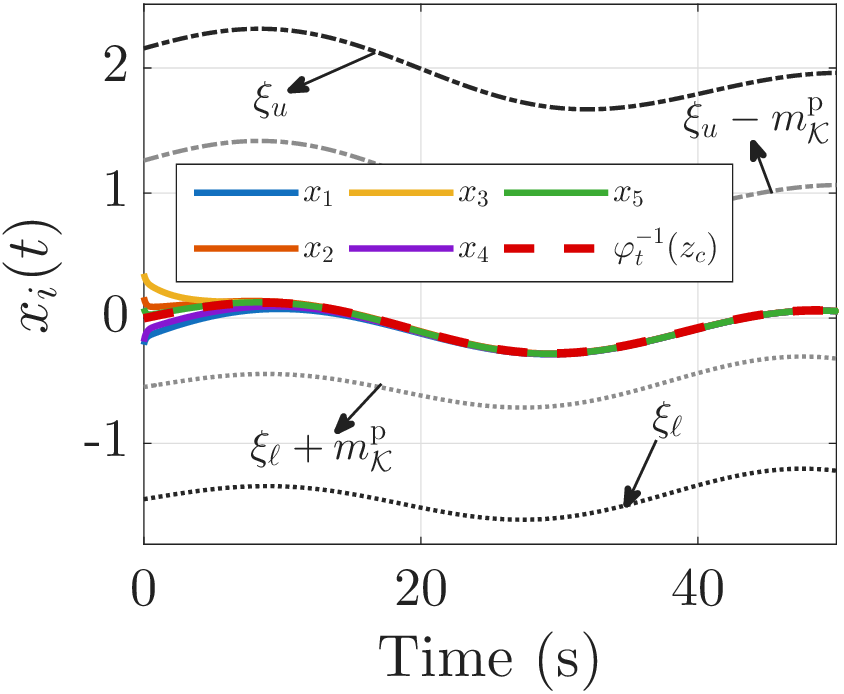}
		\caption{State trajectories $x_i(t)$.}
		\label{fig:x3}
	\end{subfigure}
	\begin{subfigure}[t]{0.48\linewidth}
		\centering
		\includegraphics[width=\linewidth]{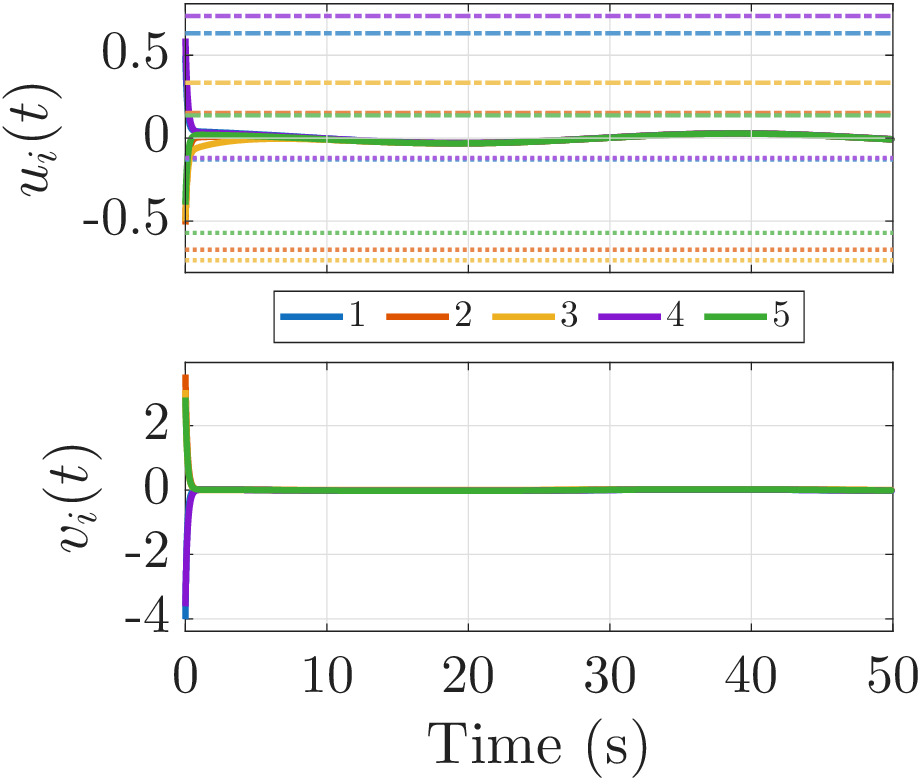}
		\caption{$u_i(t)$ and $v_i(t)$.}
		\label{fig:input_ana3}
	\end{subfigure}
	\caption{Safe consensus under partial pinning over $\mathcal{G}_2$.}
	\label{fig:pinning_G2}
\end{figure*}
\begin{figure*}[h!]
	\centering
	\begin{subfigure}[t]{0.48\linewidth}
		\centering
		\includegraphics[width=\linewidth]{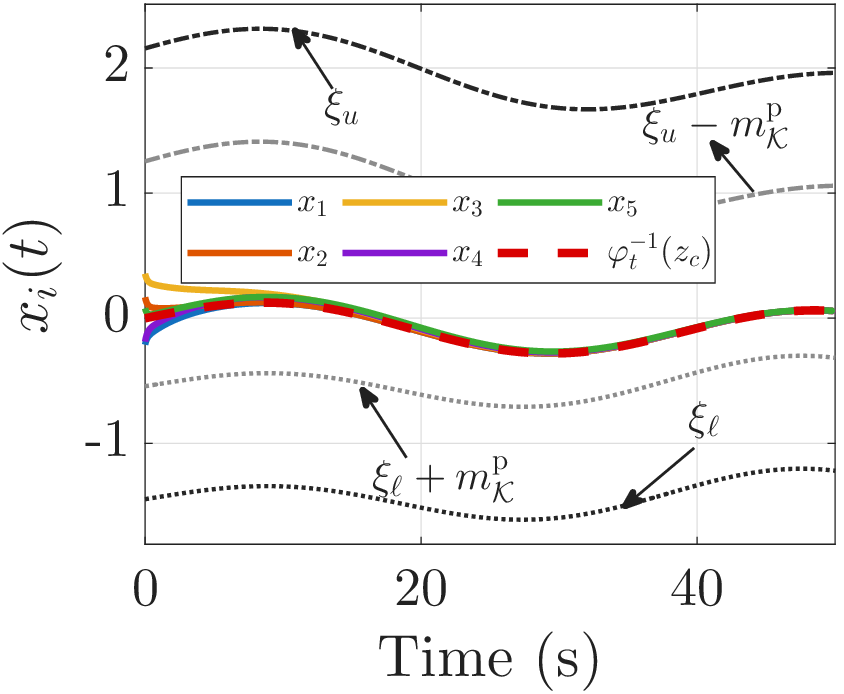}
		\caption{State trajectories $x_i(t)$.}
		\label{fig:x4}
	\end{subfigure}
	\begin{subfigure}[t]{0.48\linewidth}
		\centering
		\includegraphics[width=\linewidth]{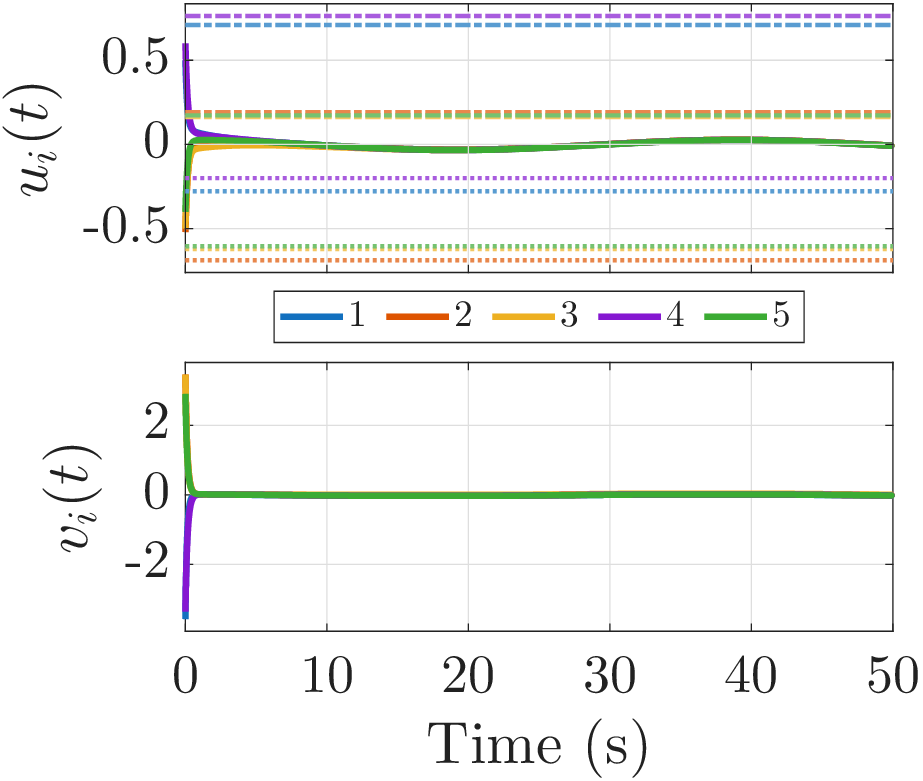}
		\caption{$u_i(t)$ and $v_i(t)$.}
		\label{fig:input_ana4}
	\end{subfigure}
	\caption{Safe consensus under partial pinning over $\mathcal{G}_3$.}
	\label{fig:pinning_G3}
\end{figure*}

To illustrate \Cref{cor:pinning} on rooted digraphs that are not strongly connected, we next consider partial pinning over the digraphs $\mathcal{G}_2$ and $\mathcal{G}_3$ in \Cref{fig:graphs}. The design parameters, actuator intervals, corridor, and initial conditions are those of the partial-pinning case on $\mathcal{G}_1$, so that only the digraph and the informed set change. The digraph $\mathcal{G}_2$ in \Cref{fig:spanning_tree} has the single-agent root component $\mathcal{R}=\{2\}$, and we set $\mathbf{G}=\diag\{0,1,0,0,0\}$, so that $\mathcal{P}=\mathcal{R}$. The digraph $\mathcal{G}_3$ in \Cref{fig:spanning_tree_G3} also contains a directed spanning tree but is not strongly connected, and its root component $\mathcal{R}=\{1,2\}$ consists of two agents. We set $\mathbf{G}=\diag\{1,0,0,0,0\}$, so that agent~2 belongs to $\mathcal{R}$ but is not informed. Both the cases satisfy $\mathcal{P}\cap\mathcal{R}\neq\varnothing$. For $\mathcal{G}_2$, the actuator margins $(\mu_i^{+},\mu_i^{-})$ are $(0.567,0.812)$, $(1.208,0.827)$, $(1.206,0.714)$, $(0.823,1.591)$, and $(1.352,0.879)$ for agents $1$--$5$, and, for $\mathcal{G}_3$, they are $(0.490,0.662)$, $(1.169,0.812)$, $(1.377,0.829)$, $(0.797,1.510)$, and $(1.317,0.846)$. As shown in \Cref{fig:pinning_G2,fig:pinning_G3}, all agents remain strictly within the corridor and track the reference $x_{\mathrm{c}}(t)$, while the realized inputs remain within their actuator intervals and the commands $v_i(t)$ remain bounded. At $t=50$~s, $\max_i|z_i(50)-z_{\mathrm{c}}|\approx1.37\times10^{-3}$ for $\mathcal{G}_2$ and $2.81\times10^{-3}$ for $\mathcal{G}_3$. These are practically zero. Note that, in $\mathcal{G}_3$, a single informed agent in the root component suffices. Therein, agent~2 receives $z_{\mathrm{c}}$ through agent~1, and the remaining agents receive it through the directed spanning tree.

The proposed synthesis is further validated on the problem of planar formation tracking with constant physical offsets. Each agent now has the position $\mathbf{x}_i=[x_i^{1},x_i^{2}]^{\top}\in\R^{2}$ and the dynamics $\dot{x}_i^{a}=u_i^{a}$, $a\in\{1,2\}$, where each component $u_i^{a}$ is generated by an APIR \eqref{eq:apir} with its own command $v_i^{a}$ and actuator interval $\cU_i^{a}=(\underline{u}_i^{a},\overline{u}_i^{a})$. The safe region is the moving box $\Omega^{1}(t)\times\Omega^{2}(t)$, where each $\Omega^{a}(t)=(\xi_{\ell}^{a}(t),\xi_u^{a}(t))$ satisfies \Cref{ass:corridor} with the constants $\delta_{\xi}^{a}$, $\Delta_{\xi}^{a}$, and $\nu_{\ell}^{a,\pm},\nu_u^{a,\pm}$. Let $\mathbf{r}_i=[r_i^{1},r_i^{2}]^{\top}$ be constant offsets, so that the formation objective is $\mathbf{x}_i(t)-\mathbf{x}_j(t)\to\mathbf{r}_i-\mathbf{r}_j$. For each axis, let $r_{\min}^{a}:=\min_i r_i^{a}$, $r_{\max}^{a}:=\max_i r_i^{a}$, the span $\bar{r}^{a}:=r_{\max}^{a}-r_{\min}^{a}$, and the shifted outputs $\tilde{x}_i^{a}:=x_i^{a}-r_i^{a}$, which satisfy $\dot{\tilde{x}}_i^{a}=u_i^{a}$.

All shifted outputs on axis $a$ share the corridor
\begin{equation}\label{eq:formation_corridor}
\tilde{\xi}_{\ell}^{a}(t):=\xi_{\ell}^{a}(t)-r_{\min}^{a},~~
\tilde{\xi}_u^{a}(t):=\xi_u^{a}(t)-r_{\max}^{a},
\end{equation}
of width $h^{a}(t)-\bar{r}^{a}$. A common corridor is essential since with agent-specific corridors of equal width, agreement in barrier coordinates would cancel the offsets. If $\tilde{x}_i^{a}\in(\tilde{\xi}_{\ell}^{a},\tilde{\xi}_u^{a})$, then
\begin{equation}\label{eq:formation_slot}
\xi_{\ell}^{a}+(r_i^{a}-r_{\min}^{a})<x_i^{a}<\xi_u^{a}-(r_{\max}^{a}-r_i^{a}),
\end{equation}
so each agent remains in its own slot within $\Omega^{1}(t)\times\Omega^{2}(t)$. For $\bar{r}^{a}<\delta_{\xi}^{a}$, the shifted corridor \eqref{eq:formation_corridor} satisfies \Cref{ass:corridor} with the width bounds $\delta_{\xi}^{a}-\bar{r}^{a}$ and $\Delta_{\xi}^{a}-\bar{r}^{a}$ and unchanged velocity and acceleration bounds, because the shifts are constant. The condition $\bar{r}^{a}<\delta_{\xi}^{a}$ is a uniform sufficient condition ensuring that the bounding box of the formation fits inside the safe box for all admissible corridor widths.

Since the dynamics, the box constraints, the actuator intervals, and the offsets are decoupled across axes, each axis is an instance of the scalar problem over the same digraph. On axis $a$, the barrier coordinate \eqref{eq:barrier}, the field \eqref{eq:beta} or \eqref{eq:pinned_beta}, the error \eqref{eq:e}, and the command \eqref{eq:controller} are evaluated with $(\tilde{x}_i^{a},\tilde{\xi}_{\ell}^{a},\tilde{\xi}_u^{a},u_i^{a})$, that is, $z_i^{a}=\ln\big((\tilde{x}_i^{a}-\tilde{\xi}_{\ell}^{a})/(\tilde{\xi}_u^{a}-\tilde{x}_i^{a})\big)$. To form the shifted neighbor states, agent $i$ requires $\mathbf{x}_j$ and $\mathbf{r}_j$, or directly $\tilde{\mathbf{x}}_j$, as well as with $\mathbf{u}_j$ from each in-neighbor $j$; the constants $r_{\min}^{a}$ and $r_{\max}^{a}$ are needed only to construct the common shifted corridor \eqref{eq:formation_corridor}. \Cref{prop:cascade} then yields \eqref{eq:cascade}, or \eqref{eq:pinned_reduced} under pinning, separately for each axis. With $M_{\cK}^{a}$ and $E_{i,\cK}^{a,\pm}$ evaluated from the shifted initial data on axis $a$, the corridor margins and directional demands become $m_{\cK}^{\mathrm{f},a}:=(\delta_{\xi}^{a}-\bar{r}^{a})/(1+e^{M_{\cK}^{a}})$ and
\begin{equation}\label{eq:formation_demand}
U_{i,\cK}^{\mathrm{f},a,\pm}:=\max\{\nu_{\ell}^{a,\pm},\nu_u^{a,\pm}\}+\frac{\Delta_{\xi}^{a}-\bar{r}^{a}}{4}\left(2kd_i^{\mathrm{in}}M_{\cK}^{a}+E_{i,\cK}^{a,\pm}\right).
\end{equation}
If $U_{i,\cK}^{\mathrm{f},a,+}<\overline{u}_i^{a}$ and $U_{i,\cK}^{\mathrm{f},a,-}<-\underline{u}_i^{a}$ for every agent and both axes, then \Cref{thm:main} applies to each axis. Each agent remains in its slot \eqref{eq:formation_slot} with the margin $m_{\cK}^{\mathrm{f},a}$, each input component satisfies $-U_{i,\cK}^{\mathrm{f},a,-}\leq u_i^{a}(t)\leq U_{i,\cK}^{\mathrm{f},a,+}$, all commands are bounded, and
\begin{equation}\label{eq:formation_limit}
x_i^{a}(t)-\big(\tilde{x}_{\mathrm{c}}^{a}(t)+r_i^{a}\big)\longrightarrow0,~~
\tilde{x}_{\mathrm{c}}^{a}(t)=\frac{\tilde{\xi}_{\ell}^{a}(t)+\tilde{\xi}_u^{a}(t)e^{z_{\infty}^{a}}}{1+e^{z_{\infty}^{a}}},
\end{equation}
exponentially, with $z_{\infty}^{a}$ given by \eqref{eq:zinf} on axis $a$. Hence, $\mathbf{x}_i(t)-\mathbf{x}_j(t)\to\mathbf{r}_i-\mathbf{r}_j$, and the asymptotic inter-agent separation equals $\|\mathbf{r}_i-\mathbf{r}_j\|$. The offsets $\mathbf{r}_i$ thus determine the formation shape, whereas the common trajectory $\tilde{\mathbf{x}}_{\mathrm{c}}(t)=[\tilde{x}_{\mathrm{c}}^{1}(t),\tilde{x}_{\mathrm{c}}^{2}(t)]^{\top}$ determines its translation. Under root pinning with $\mathcal{P}\cap\mathcal{R}\neq\varnothing$ and the shifted demands from \eqref{eq:pinned_demand}, constant references $z_{\mathrm{c}}^{a}$ replace $z_{\infty}^{a}$ in \eqref{eq:formation_limit}, so that the informed agents assign the planar trajectory of the entire formation.

For greater visual clarity, please see supporting animations at \href{https://youtu.be/Td-KoCb8zkQ}{https://youtu.be/Td-KoCb8zkQ}, which include an application to formation tracking discussed above.

\section{Conclusions and Future Work}\label{sec:conclusion}
We formulated a directed safe consensus as a networked APC problem. APIR retains the physical input as a constrained dynamic state, and the distributed synthesis realizes the barrier-coordinate consensus field through the admissible input dynamics. The exact cascade preserves the rooted-digraph geometry and quantifies the APIR-transient correction to the directed-network agreement value. APC compatibility uses transformed-state bounds to certify uniform output margins, direction-specific input bounds, a positive APIR gain margin, bounded commands, and complete solutions. Partial pinning assigns the safe collective trajectory without a global reference broadcast. The regional character of the certificate reflects the finite authority available relative to the initial transformed dispersion and corridor motion. Future work includes coupled polyhedral state--input constraints, with joint facet feasibility and appropriate information exchange, and prescribed command limits, including saturation effects.
\bibliographystyle{IEEEtran}
\bibliography{references}
\end{document}